\documentclass[journal]{IEEEtran}

\usepackage{amsthm}
\usepackage{amsmath} %
\usepackage{mathtools}
\usepackage{graphicx}
\usepackage{psfrag}
\usepackage{epsf,epsfig}
\usepackage{epstopdf}
\usepackage{subfigure}
\usepackage{cite,color}
\usepackage{cite}
\usepackage{amsmath,amssymb,amsfonts}
\usepackage{algorithmic}
\usepackage{graphicx}
\usepackage{textcomp}
\usepackage{xcolor}
\usepackage{verbatim}

\definecolor{red}{rgb}{1,0,0}  

\newcommand{\expect}[1]{\operatorname{E}\left[#1\right]}   

\newcommand{\Bcal}{\mathcal{B}}
\newcommand{\Ccal}{\mathcal{C}}

\newcommand{\Kcal}{\mathcal{K}}
\newcommand{\Ncal}{\mathcal{N}}

\newcommand{\bbf}{{\bf b}}

\newcommand{\Fbf}{{\bf F}}

\newcommand{\Ibf}{{\bf I}}

\newcommand{\xbf}{{\bf x}}

\newcommand{\Nbit}{N_{\rm bit}}
\newcommand{\Nbo}{N_{\rm bo}}
\newcommand{\Npl}{N_{\rm pulse}}  
\newcommand{\Nsymb}{N_{\rm pulse}}  
\newcommand{\Nseg}{N_{\rm seg}}

\newcommand{\Nsc}{N_{\rm sc}}
\newcommand{\Nfft}{N_{\rm fft}}
\newcommand{\Ncp}{N_{\rm cp}}

\newcommand{\jrm}{{\rm j}}

\newtheorem{Cor}{Corollary}

\newtheorem{Lem}{Lemma}

\begin{document}

\title{DFT-s-OFDM-based On-Off Keying \\for Low-Power Wake-Up Signal}
\author{
\IEEEauthorblockN{Renaud-Alexandre Pitaval and Xiaolei Tie}
\thanks{
Renaud-Alexandre Pitaval is with Huawei Technologies Sweden AB, Kista 164 94, Sweden (e-mail: renaud.alexandre.pitaval@huawei.com). 

Xiaolei Tie is with Shanghai Huawei Technologies Co., Ltd, NO.1599, Xin Jinqiao Road, Pudong, Shanghai, P.R.C (e-mail:tiexiaolei@huawei.com).

© 2025 IEEE.  Personal use of this material is permitted.  Permission from IEEE must be obtained for all other uses, in any current or future media, including reprinting/republishing this material for advertising or promotional purposes, creating new collective works, for resale or redistribution to servers or lists, or reuse of any copyrighted component of this work in other works.
}
}

\maketitle

\begin{abstract} 5G-Advanced and likely 6G will support a new low-power wake-up signal (LP-WUS) enabling low-power  devices, equipped with a complementary ultra low-power receiver to monitor wireless traffic, to completely switch off their main radio. This orthogonal frequency-division multiplexed  (OFDM)  signal will emulate an on-off keying (OOK) modulation to enable very low-energy envelope detection at the receiver. Higher rate LP-WUS, containing multiple OOK symbols within single OFDM symbol, will be generated using the time-domain pulse multiplexing of discrete Fourier transform spread (DFT-s-) OFDM. 
In this context, this paper presents a comprehensive signal design framework for DFT-s-OFDM-based OOK generation. 
General properties of subcarrier coefficients 
are derived demonstrating that only DFT of the bits needs to be computed online and repeated over the band before applying appropriate frequency-domain processing.  
The conventional approach of generating rectangular-like OOK waveforms is then addressed by a combination of pre-DFT bit-spreading and post-DFT processing; and the  least-squares (LS)   method from Mazloum and Edfors, proposed for 5G LP-WUS  and also Ambient-IoT, is shown to be implementable as such.  
Even though aesthetically pleasing and of independent interest,  rectangular-like OOK waveforms are not optimal for 5G LP-WUS scenarios
due to their limited robustness to channel frequency-selectivity and timing offset, and so shaping methods for spreading the OOK spectrum and concentrating the OOK  symbol energy are analyzed  and shown to improve the bit error rate performance under practical conditions.

\end{abstract}

\begin{IEEEkeywords}
OOK, DFT-s-OFDM, wake-up signal, Ambient-IoT, 5G-Advanced, 6G.
\end{IEEEkeywords}

\section{Introduction}
The Internet-Of-Things (IoT), consisting of low-power devices, is generating renewed interest in easily decodable modulation schemes such as on-off keying (OOK). 
The concept of a wake-up signal (WUS) has been introduced in several standards within the wireless industry. 
The goal is to enable devices to significantly reduce their functionality, and thus their power consumption, until  incoming traffic is indicated by the network.    
 
In 4G LTE-M and NB-IoT  technologies, 3GPP already defined a WUS to enable receiver sleep modes~\cite{EricssonMag24}, 
but its reception relies on a standard OFDM radio, which still demands precise synchronization and high-precision analog-to-digital conversion (ADC). 
Recently, 3GPP RAN1 completed  a study~\cite{3GPPTR38.869} on a so-called low-power WUS (LP-WUS)  
with a first specification for 5G-Advanced. 
Recent works motivated by this study can be found in~\cite{Eurecom24,FredrikVTC24,AlbertoICC24,Ericsson23,EricssonSoft24}. 
Significantly larger power-saving gains~\cite{MazloumTWC14} are envisioned  if the main radio (MR) of a device could be completely switched off when no message is coming, and  being triggered only when necessary. 
To enable this, the device would need to be equipped with an additional lower-power radio (LR) that monitors possible incoming traffic indicated by an LP-WUS. 
To achieve ultra low-power consumption, the LR could employ a simple non-coherent envelope detector, through which an LP-WUS can convey  information using OOK modulation.  
Meanwhile, at the transmitter side, an OFDM generation of the LP-WUS is preferred to facilitate implementation and orthogonal multiplexing with concurrent transmissions. 
The design of  LP-WUS paved the way for other OFDM-based OOK-modulated signals in 3GPP, already for the OOK signal in Ambient-IoT (A-IoT)~\cite{3GPPTR38.769} and likely for other upcoming 6G IoT technologies.

\emph{Background on OOK for Wi-Fi WUS:}
Generating an OOK-modulated WUS while reusing an existing OFDM implementation has already been considered in the IEEE 802.11ba Wi-Fi standard in order to ease WUS adoption. 
The Wi-Fi WUS supports a low-rate OOK mode matching the OFDM symbol rate and a high-rate mode, doubling the rate through time-blanking within OFDM symbols. 
The latter, however, is non-orthogonal to concurrent OFDM signals, resulting in inter-carrier interference (ICI). 
In the context of Wi-Fi, non-orthogonal WUS designs have been proposed in~\cite{LopezWCNC18,LopezICC19}, and orthogonal OFDM-based designs in~\cite{SahinGC18,SahinGlobecom19,Mazloum20}.  
In~\cite{SahinGC18},  
sequences of subcarrier coefficients are numerically-found under several constraints such as ON/OFF power levels, and a zero direct current (DC) subcarrier as in legacy Wi-Fi signal format. 
In~\cite{SahinGC18}, Golay sequences are used instead to limit the high Peak-to-Average Power Ratio (PAPR) from OFDM multiplexing of multiple, but low-rate, OOK signals.  
A scalable high-rate OFDM-based approach is proposed by Mazloum and Edfors in~\cite{Mazloum20}  
by deriving a precoder, mapping spread bits to subcarrier coefficients,  such that the resulting OFDM signal minimizes the least-square (LS) error from a perfectly-rectangular OOK signal.

\emph{Background on OOK for 5G-Advanced LP-WUS:} When preparing our contributions to 3GPP standardization of LP-WUS, our initial approach  
 independently explored the inherent  time-domain (TD) pulse multiplexing property of DFT-s-OFDM 
to generate higher rate but OFDM-compliant OOK signals. We  first targeted  a rectangular-like OOK  
by controlling pulse combining along with frequency-domain (FD) spectral shaping (FDSS) as well as spectrum extension (SE), and later on investigating  other shaping techniques for performance improvement, including FD repetition, alternative bit-spreading sequences and TD pulse shaping. This is reflected in related 3GPP contributions, e.g.,~\cite{R1-2208419,R1-2300102,R1-2302341}  which discuss the LS method from~\cite{Mazloum20}, $\pm 1$ alternating bit-spreading sequence for rectangular OOK,  Zadoff-Chu (ZC) bit-spreading for improving performance in frequency-selective channels, and concentrated OOK for robustness against timing misalignment. These signal designs were  also captured in~\cite{3GPPTR38.869} as potentially useful for performance improvement.

From the perspective of wireless standard engineers, using DFT-s-OFDM to create OFDM-compliant  OOK signals is a rather natural approach, and  similar considerations were therefore  independently proposed by many 3GPP participants from the beginning of the 3GPP LP-WUS study~\cite{R1-2212749}. In retrospect, we identify  the possibly  first occurrence  of DFT-s-OFDM-based OOK in a Wi-Fi WUS standardization contribution by A. Sahin, R. Yang~\emph{et al.}~\cite{SahinDFTbasedOOK}, where, interestingly, FDSS-SE and ZC bit-spreading are already listed among possible signal design ingredients. However,  as is often the case, standardization contributions focus on addressing specific use cases with a limited scope and less detailed analysis of  signal design alternatives -- for instance the choice of a ZC sequence in~\cite{SahinDFTbasedOOK} is not further justified.

The LS processing from~\cite{Mazloum20} to obtain rectangular-like OOK,  introduced in 3GPP by us~\cite{R1-2208419} and others~\cite{R1-2212749}, was also one of the main approaches considered at the initial stage of the LP-WUS study. As reflected in~\cite{3GPPTR38.869}, this approach was seen by many participants as a specific OOK generation method, involving a specific and so-called ``LS precoding'', distinct from the DFT precoding of DFT-s-OFDM -- even though this LS precoder is simply a larger, truncated DFT matching both the OFDM sampling rate at its input and the bandwidth at its output. 
Direct FD sequence-based OFDM-OOK was also suggested in 3GPP LP-WUS study as an alternative to TD approaches. Eventually, 5G LP-WUS is a DFT-s-OFDM-based OOK signal with ZC overlaid sequence.  

\emph{Background on OOK for 5G-Advanced A-IoT:}
LP-WUS study influenced the more recent 3GPP study on A-IoT~\cite{3GPPTR38.769} for the  reader-to-device OOK waveform generation. In this context, a rectangular-like OOK signal design can be a relevant approach notably as 
 edge/transition  detection between adjacent OOK symbols is likely to be used at receiving devices, for which  OOK symbol flatness becomes a desired signal property~\cite{R1-2500349,R1-2500424}.  
Both DFT-s-OFDM-based and LS approaches  have been considered for A-IoT~\cite{R1-2405441}  and enabled by a signal generation that allows a DFT precoder larger than the subcarrier allocation, unlike LP-WUS. Other waveform generation steps are left to manufacturer-specific implementation choices.

\emph{Motivations, Benefits and Limitations:} In both LP-WUS and A-IoT, the primary motivation of using DFT-s-OFDM for OOK lies in its time-multiplexing property, while still complying with legacy OFDM-based system and hardware, and not for its traditionally considered main benefit which is  a low PAPR.  
Notably, these applications consider  OOK transmission inserted in a downlink OFDM  signal  
whose other subcarriers could carry high-order QAM symbols, resulting overall in a high PAPR. In this context, the transmitted signal as a whole  is OFDM-based and not DFT-s-OFDM-based.

When DFT-s-OFDM is used in 5G uplink data transmission, and with $\pi/2$-BPSK, it can be complemented by FDSS  to achieve very low PAPR.  For higher order QAM, FDSS with SE has been also considered~\cite{Nokia21}. The low-PAPR benefit of FDSS-SE in this context comes  at the cost of breaking  orthogonality among the DFT-s-OFDM multicarrier pulses, thus creating ICI among the data modulation symbols.  For OOK signal design, FDSS-SE can  similarly enable control over the signal's envelope, but in this case, without any interference issue as OOK detection is blind to the orthogonality property of the underlying  DFT-s-OFDM pulses. On the contrary, by limiting energy leakage among consecutive OOK symbols, FDSS-SE can help to mitigate interference among them. 
In any case, the OFDM subcarriers of the OOK signal, with or without FDSS-SE,  
remain orthogonal to those of concurrent transmissions and thus interference-free to them.

Another concern that could arise from OOK -- especially rectangular OOK -- is the potentially large  out-of-band  emission (OOBE). However, with an OFDM-based approach,  most of the power is by design concentrated within the subcarrier allocation, yielding an  OOBE level    comparable to typical OFDM signals~\cite{SahinDFTbasedOOK}. 
The OOBE of OFDM is well known to not directly satisfy  regulatory requirements, but standard filtering techniques such as root-raised-cosine filter can often address this issue effectively.
 Moreover, in practice this OOK signal may be inserted within an OFDM signal of larger bandwidth, rendering its impact   negligible on the overall OOBE.  On the other hand, an inherent limitation of the OFDM-based approach is that the  transition slopes between OOK symbols are constrained by the subcarrier allocation.

\emph{This Paper's Contributions:}  We provide a comprehensive framework for DFT-s-OFDM-based OOK modulation and study both its TD and FD characteristics. 
While primarily motivated by the 5G  LP-WUS application,  
the presented results are general and may be of independent interest for other use cases.  
The considered framework encompasses an overlaid bit-spreading sequence prior to DFT-precoding, possibly combined with FDSS and SE. The  bit-spreading sequence modulates the DFT-s-OFDM pulses, while FDSS-SE offers additional degrees of freedom for shaping the pulses and controlling their number.   

In both Wi-Fi WUS and the 5G LP-WUS study, defining subcarrier sequences directly was also considered as an alternative to DFT-s-OFDM with bit spreading. We link these two methods by showing that a common TD overlaid sequence on top of bit spreading results in an FD overlaid sequence  that creates an ON-symbol kernel whose TD multiplexing is controlled by the DFT of the bits.    
This insight reveals an equivalent low-complexity implementation, where the data-dependent DFT precoding is applied directly to the input bit string --  without the need for explicit bit spreading.  
Furthermore,  compact DFT formulas and tables that absorb Manchester encoding are derived for efficient storage. 

Next,  we focus on the design of  rectangular-like OOK waveforms, as considered in Wi-Fi WUS, the early stage of the 5G LP-WUS study, and 5G A-IoT study.  A linear phase ramp is applied in the bit-spreading sequence to minimize the coherent combining of consecutive DFT-s-OFDM pulses, thereby flattening the ON-symbol envelope.  
We also demonstrate how the bit-spreading sequence can be chosen  to obtain a rectangular-like OOK waveform  with the additional feature of a zero DC component. Moreover, we show that the LS waveform from~\cite{Mazloum20} arises  
as a special case of the considered DFT-s-OFDM-based OOK modulation, using a specific FDSS window.

Finally, we move beyond the intuitive approach of regular and rectangular-like OOK waveform  to propose and analyze shaping methods  better suited for  LP-WUS application scenarios.    
Rectangular-like OOK designs, though relevant to small indoor coverage applications such as Wi-Fi and A-IoT, are not optimal for 5G LP-WUS.  Large cell coverage, as supported by the 5G NR standard, involves larger bandwidth usage and increased sensibility to   
frequency-selectivity and channel time dispersion. Additionally, 3GPP considered a significantly higher sampling rate at LR compared to the transmitted OOK symbol rate.  In this context, as long as the OOK signal achieves a sufficient ON/OFF energy split, envelope fluctuations actually matter little.

We formalize the effect on the power spectrum distribution obtained from phase scrambling by an ZC overlaid sequence, as considered in~\cite{R1-2302213} to flatten the spectrum. More precisely, we show that this results in a constant spectral comb over the bandwidth, containing half the total signal power and thereby providing excellent frequency diversity.  
Additionally, we investigate the TD concentration of ON-symbol energy to enhance robustness against timing errors. Simulation results confirm the performance benefits of such shaping methods over rectangular OOK in LP-WUS application. 

\emph{Organization:} The remainder of the paper is as follows. Section II introduces the LP-WUS system model from an OFDM transmitter to an envelope detector at the LR. Section III  describes the OOK signal generation framework based on DFT-s-OFDM modulation highlighting its general TD and FD properties.  Section IV addresses the generation of rectangular-like OOK waveform, while Section V discussed shaping methods relevant for mitigating wireless channel impairments. Performance evaluations are presented in Section VI, and conclusions are drawn in Section VII.

\section{LP-WUS System Model}
We consider an OOK-modulated LP-WUS multiplexed along other data using a conventional OFDM transmitter, thus enabling concurrent transmission 
without interference.

\subsection{OFDM Modulation}
The LP-WUS is generated by populating $\Nsc$ subcarriers within other data symbols, all multiplexed by OFDM using an $\Nfft$-point IFFT followed by cyclic prefix (CP) addition. While such setup avoids interference to concurrent data received by OFDM terminals,  the WUS  itself is intended for a non-OFDM, envelope-detection receiver for which  concurrent data may interfere, and thus some guard bands may still be necessary. The total LP-WUS bandwidth allocation is thus $(\Nsc + 2 N_{\rm GB})$ where $N_{\rm GB}$ are null subcarriers on each band side. 

Formally, the transmitted baseband OFDM signal $s[n]$  is a superposition of the WUS  $s^W [n]$ and concurrent data signal $s^D [n]$ where one CP-OFDM symbol with sample indices $-\Ncp\leq n \leq \Nfft-1$ can be expressed as 
\begin{equation}  
s[n] =\sum_{k=0}^{\Nfft-1} X[k] e^{\jrm \frac{ 2\pi}{\Nfft}nk} =  s^W[n] + s^D [n] 
\end{equation}
and the LP-WUS is
\begin{equation} \label{eq:s^W[n]} 
s^W[n]= e^{\jrm  \frac{2\pi}{\Nfft} n f_0} \sum_{k=0}^{\Nsc-1} X^W [k] e^{\jrm \frac{ 2\pi}{\Nfft}nk} , 	
\end{equation}
i.e. it is the $\Nfft$-point inverse DFT (IDFT), typically implemented as a fast Fourier Transform (FFT), of the $\Nsc$  WUS subcarrier coefficients $\{X^W [k]\}_{k=0}^{\Nsc-1} =\{X [k]\}_{k=f_0}^{f_0+\Nsc-1}  $ allocated at the the subcarrier indices $\{f_0,\ldots,f_0+\Nsc-1\}$.

The LP-WUS is a complex signal where information is conveyed through the fluctuation of the envelope’s amplitude; thus its starting subcarrier index $f_0$ is irrelevant from signal design perspective as it only changes its global phase. For convenience, we will assume $f_0=0$ in analytical derivations.  The value of $\Nsc$ will often be assumed even for simplicity and practical relevance. For numerical evaluation and illustration, $\Nfft=512$ will be used. 

A CP protects OFDM symbols from interfering on each other after channel dispersion.  From the OOK signal perspective only, CP insertion is mainly detrimental, with its presence and length imposed on LP-WUS for system compliance.   
The CP reduces the OOK transmission rate while  offering no interference protection when more than one OOK symbol is transmitted per OFDM symbol. 
To avoid this rate reduction, the targeted OOK signal is punctured in~\cite{Mazloum20}, thereby  maintaining a regular bit period even after CP insertion, rendering LP-WUS synchronization independent of the OFDM structure, but at the cost of increased errors in bit detection.  
Moreover, CPs can introduce false rising/falling edges --  a concern raised  in the A-IoT study~\cite{3GPPTR38.769}. In both LP-WUS and A-IoT, the main  CP handling approach assumes that the receiving device  removes it without requiring specific transmitter-side assistance. 
For high-rate OOK, where CP becomes comparable to or longer than an OOK symbol, receiver-only handling may become insufficient, and transmitter-side handling -- such as reserving the last OOK symbols  to control the CP, as  discussed in~\cite{3GPPTR38.769} -- may also be needed.

\begin{figure*}[t] 
\centering
	 \vspace{-0.2cm} 
\includegraphics[width=0.75\textwidth]{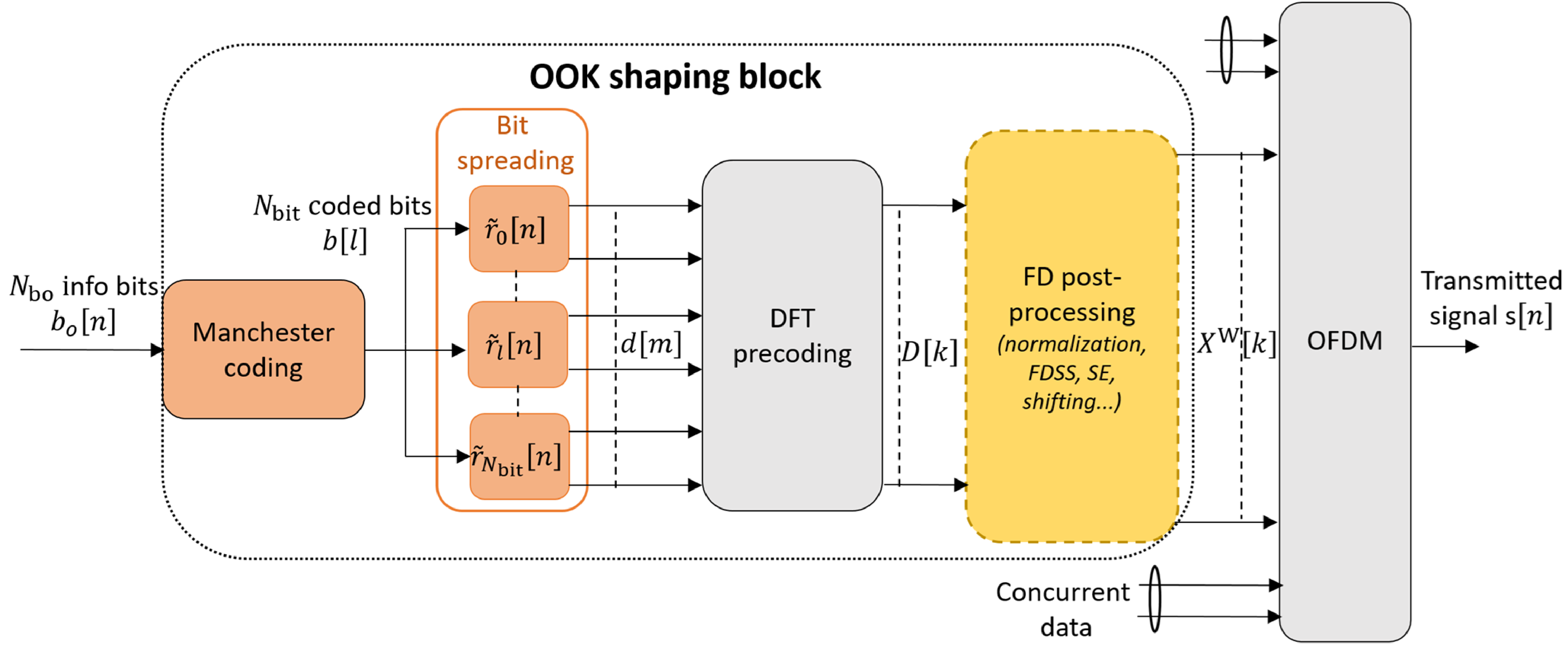} 
	 \vspace{-0.2cm}
	\caption{Illustration of DFT-s-OFDM based OOK modulation}
	\label{fig:Tx}
	 \vspace{-0.2cm}
\end{figure*}

\subsection{Envelope Detection}
The signal arrives at the receiver via a wireless channel composed of $L_p$ paths assumed with 
Rayleigh-faded attenuation $h_p \sim \Ccal\Ncal(0,E_p) $ as
\begin{equation}
y[n] = \sum_{p=0}^{L_p-1} h_p s[n-p] + z[n]  
\end{equation}
where $z[n]$ is a zero-mean complex additive white Gaussian noise (AWGN) with variance $N_0$.
 
We consider a LR as shown in Fig.~\ref{fig:LR}. The received analog signal is first passed through a bandpass filter (BPF) centered around the LP-WUS band to remove adjacent channels; followed by an envelope detector which consists of a norm operator followed by a low-pass filter to smooth the signal,  and then an ADC to finally perform detection.   We will assume a Manchester encoded OOK signal, so detection can be performed by directly comparing the energy between two consecutive OOK symbols.

The receiver is considered to capture the WUS and noise power over the WUS bandwidth, including guard bands, such that the SNR is defined as \mbox{$\mathsf{snr}= \frac{E_h P_W}{(\Nsc +2N_{\rm GB})N_0}$} where $E_h = \sum_{p=0}^{L_p-1}E_p$ is total average channel energy, and $P_W = \expect{|s^W[n]|^2}=  \sum_{k=0}^{\Nsc-1} \expect{ |X^W[k]|^2} $ is the average transmitted power of the WUS signal. 

\begin{figure}[t] 
\centering
\vspace{-0.3cm}
\includegraphics[width=0.49\textwidth]{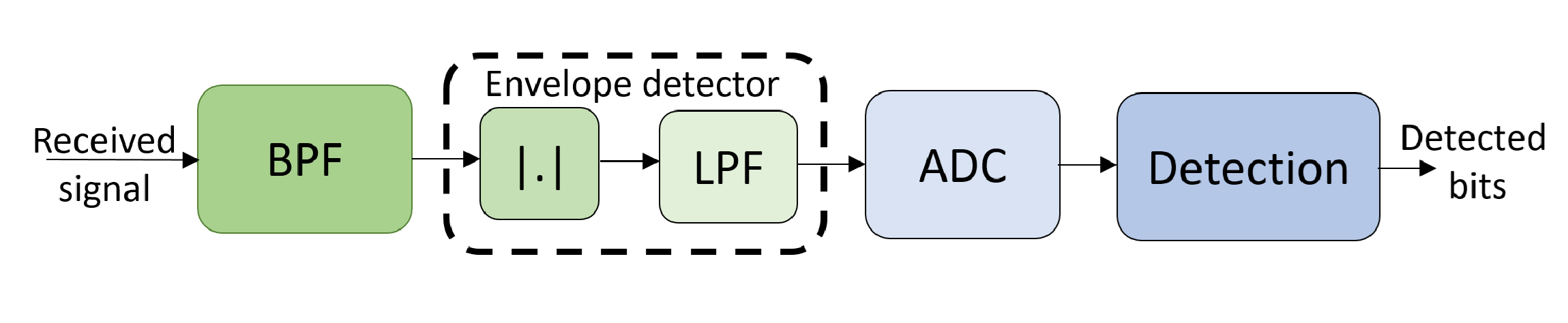}
\vspace{-0.8cm}
	\caption{A low-power wake-up receiver (LR) based on envelope detection}
	\label{fig:LR}
	\vspace{-0.4cm}
\end{figure}

\section{Generation Framework for DFT-s-OFDM-based OOK}
In each OFDM symbol, the LP-WUS OOK signal contains $\Nbit$ OOK symbols (ON or OFF), identified as $\Nbit$ coded bits  being the line coding of $\Nbo$ info bits.   
A generic method to achieve a multi-bit OFDM-OOK for LP-WUS, i.e. $\Nbit>1$ bits per OFDM symbol,  is to obtain the subcarrier coefficients  
by DFT of a string of  $\Nsymb \leq \Nsc $  
 modulation symbols that embeds the targeted ON/OFF pattern, resulting in a DFT-s-OFDM modulation  \emph{only} for $s^W[n]$,  \emph{but not} for the remaining part of the OFDM signal $s^D[n]$. 
As a result, the $\Nsymb$  input symbols then modulates or deactivates a set of selected DFT-s-OFDM time pulses within the OFDM symbol.

\subsection{DFT-s-OFDM-based OOK} 
The considered DFT-s-OFDM based OOK generation framework for LP-WUS is schematized in Fig.~\ref{fig:Tx}.
\subsubsection{Input Bits and Manchester Coding} 
The ON/OFF pattern in the OFDM symbol is controlled by the coded bit string $b[l]\in \{0,1\}$, $l=0,\ldots,\Nbit-1,$ which we will  consider (most of the time) to be  
the Manchester encoding of an original info bits string  $b_o[n]\in \{0,1\}$, $n=0,\ldots,\Nbo-1$  of length $\Nbo= \Nbit/2$.  This Manchester encoding can be formally written as 
\begin{equation} \label{eq:MC}
b_o [n]\to (b[2n],b[2n+1])=( \overline{b_o [n]}, b_o[n]),
\end{equation}
where $\overline{b}= (1 + b) \,{\rm mod}\, 2$ is the NOT (bit-flipping) operation\footnote{Remark that Manchester coding can also be defined as $(b[2n],b[2n+1])=( b_o[n], \overline{b_o [n]})$, but~\eqref{eq:MC} corresponds to a more common implementation obtained  by XOR of the clock signal and the bit string.}.  

\subsubsection{Bit Spreading}
The bits $b[l]$ are spread by an integer factor $\Nseg =\Nsymb/ \Nbit$ to obtain a sequence of symbols $d[m]$ that will be used to modulate  $\Nsymb \leq \Nsc $ DFT-s-OFDM pulses. 
If a bit $b[l]=0$ then it is mapped to an all-zero sequence; otherwise if  $b[l]=1$ it is mapped to a (so-called in 3GPP) overlaid sequence $\tilde{r}_l[n]$, $n=0,\ldots,\Nseg-1$. In general, each bit $b[l]$ can be written as $\Nseg$-time repeated  as $b_{\rm r}[m] =b[l_m ] $ with $l_m = \lfloor \frac{m}{\Nseg} \rfloor$ 
and element-wise multiplied with a concatenated $\Nsymb$-long bit-spreading sequence  $r[m] =  \tilde{r}_{l_m} [m \; {\rm mod  }\; \Nseg ]$. The modulation symbols  are thus
\begin{eqnarray}
d[m]&=& b_{\rm r}[m]r[m] \quad \quad \quad  \quad \text{ for }  m=0,\ldots,\Nsymb-1 \nonumber \\
&=& b[l_m ]  \tilde{r}_{l_m} [m \; {\rm mod  }\; \Nseg ]. \label{eq:d[m]}
\end{eqnarray}

\subsubsection{DFT-precoding}
The modulation symbols are then DFT-precoded to obtain a set of pre-processed subcarrier coefficients  with  indexes $0\leq k \leq \Npl-1$ as \begin{equation} \label{eq:D[k]}
D[k] = \sum_{m=0}^{\Nsymb-1}d[m]  e^{- j \frac{2\pi}{\Nsymb} km}.  
\end{equation}

\subsubsection{FD Post-Processing}  
Typical FD processing of DFT-s-OFDM~\cite{Nokia21} are also considered here. The $\Nsymb$ coefficients $D[k]$ may be cyclically-extended as well as shaped by an FDSS window before mapping to the $\Nsc$ WUS subcarriers, such that for $k =0,\ldots, \Nsc-1$,
\begin{equation} \label{eq:X^W[k]}
X^W[k] = \eta W[k] D[(k+L) \; {\rm mod  }\; \Nsymb]  
\end{equation}
where $(\Nsc-\Nsymb)$ is the spectrum extension (SE), $L$ is an integer shift, $W[k]$ is an FDSS window, and $\eta$ is a normalization factor.  Conventional DFT-s-OFDM   is with $\Nsymb=\Nsc$, $L=0$,  and then $X^W[k] = \eta D[k]$.

SE  enables to control the number of pulses  $\Nsymb \leq \Nsc $  
such that the spreading factor $\Nseg =\Nsymb/ \Nbit $ is an integer. For example, with $\Nbit = 8$ and bandwidth $\Nsc = 12\times 11 = 132$, $\Nbit/\Nsc$ is not an integer; so instead by selecting $\Nsymb = 128$ one gets an integer spreading factor $\Nseg = \Nsymb/\Nbit$. 

FDSS enables shaping of the ON symbols. We consider an FDSS window 
limited to the form $W[k] = e^{-\frac{\jrm 2 \pi }{\Nfft} T_{\rm shift} k} W_R[k]$ where $W_R[k]$ is a real and symmetric FDSS window and $ T_{\rm shift}$ corresponds to a cyclic time shift of the waveform. 	For numerical evaluations, because it concentrates well the energy of DFT-s-OFDM pulses~\cite{Mauritz06},   we will use a Kaiser window defined as $W_R[k] = I_0\left( \beta \sqrt{1-\frac{(k-\gamma)^2}{\gamma^2}}\right)/  I_0\left( \beta \right)$  with shaping parameter $\beta$  and where $\gamma = (\Nsc-1)/2$ and $I_0(\cdot)$ represents the zeroth-order modified Bessel function of the first kind.  With $\beta=0$, $W_R[k]=1$ for all $k$ i.e. there is no shaping; and as the parameter $\beta$ increases, the FD energy concentration increases.

The normalization factor  $\eta$ adjust the LP-WUS power.  
In~\cite{3GPPTR38.869}, LP-WUS power spectral density should take into account the guard band usage to be aligned with the concurrent average symbol power $P_s$, i.e., $\eta$ in~\eqref{eq:X^W[k]} is selected  satisfying  
$\sum_{k=0}^{\Nsc-1} \expect{ |X^W[k]|^2}  =P_s(\Nsc +2N_{\rm GB}) $.

\subsection{TD Pulse Multiplexing Aspect}
The inherent TD multiplexing of DFT-s-OFDM makes it suitable for emulating OOK modulation.  
After inserting \eqref{eq:D[k]} and \eqref{eq:X^W[k]} in \eqref{eq:s^W[n]}, the WUS can be expressed as 
\begin{equation} \label{eq:sW_gm}
s^W[n]= \eta \sum_{m=0}^{\Nsymb-1}d[m]  g_m[n],
\end{equation}
which is the multiplexing of symbols $d[m]$ by $\Nsymb$ pulses 
\begin{equation} \label{eq:gm}
 g_m[n]=  e^{-\jrm \frac{2 \pi L}{\Nsymb}m} h\left[n-\frac{\Nfft}{\Nsymb} m \right],
\end{equation}
all being time-shifted versions of the same kernel filter
\begin{equation}  \label{eq:PulseKernel}
 h[n]=  \sum_{k=0}^{\Nsc-1} W[k] e^{\jrm \frac{2 \pi k}{\Nfft}n }.
\end{equation}
This filter is the IDFT of the FDSS window. Without FDSS, i.e. all $W[k]=1$, this reduces to the classical Dirichlet kernel 
\begin{equation}  \label{eq:Dirichlet}
 h[n]= e^{\jrm \frac{\pi (\Nsc-1) }{\Nfft} n} \frac{\sin\left ( \pi \frac{\Nsc}{\Nfft}n \right)}{\sin\left ( \pi \frac{1}{\Nfft}n \right)}.
\end{equation} 
By inserting~\eqref{eq:d[m]} in \eqref{eq:sW_gm}, the resulting TD bit multiplexing can be written as 
\begin{equation}\label{eq:s[n]Op}
s^W[n] 
= \eta \sum_{l=0}^{\Nbit-1} b[l] O_l[n]
\end{equation}
where 
\begin{equation}
O_l[n] =\!\!\! \sum_{m=l \Nseg}^{(l+1)\Nseg-1} \!\!\!\!\!\!  r[m] g_m[n]
= \sum_{m=0}^{\Nseg-1} \!\!\! \tilde{r}_l[m] g_{m+l\Nseg}[n] \label{eq:Ol[n]}
\end{equation}
is the OOK symbol for the $l$th bit.  Each pulse $g_m [n]$ has most of its energy confined in the sample interval $n\in \left[m\frac{\Nfft}{\Nsymb}  , (m+1)\frac{\Nfft}{\Nsymb}  \right]$ with an energy peak in the middle. Hence, most of the potential combined energy of $O_l[n] $ is in the sample interval 
$n\in \left[\frac{\Nfft}{\Nbit}l,\ldots,\frac{\Nfft}{\Nbit}(l+1)\right]$ suitable for creating an ON symbol in it, while combined energy leakage from these pulses in other sample intervals can contribute to undesired signal level in OFF symbols.    
Fig.~\ref{fig:pulses} illustrates this pulse modulation aspect within one OFDM symbol. Remark in Fig.~\ref{fig:pulses}(a) that the first pulse is always peaking at time zero with its energy spread equally at the beginning and end of the OFDM symbol. This will result in a circular leakage from the first OOK symbol to the last OOK symbol, but  can be attenuated by shifting all pulses to the right using a linear phase in $W[k]$ of  $T_{\rm shift}= \frac{\Nfft}{2\Nsymb}$ as in  Fig.~\ref{fig:pulses}(b). 
Also in Fig.~\ref{fig:pulses}(b), SE is applied, which decreases the number of pulses and spread them apart; and the pulse kernel~\eqref{eq:PulseKernel} is shaped with $\beta=2$ which has the typically effect  to enlarge its main lobe while attenuating it sidelobes. As a result of both SE and FDSS the pulses are non-orthogonal. 
 
\begin{figure}
\vspace{-0.0cm}
\subfigure[$\Nsymb= 12$, plain DFT-s-OFDM \label{fig:pulses1}]{\includegraphics[width=.49\textwidth]{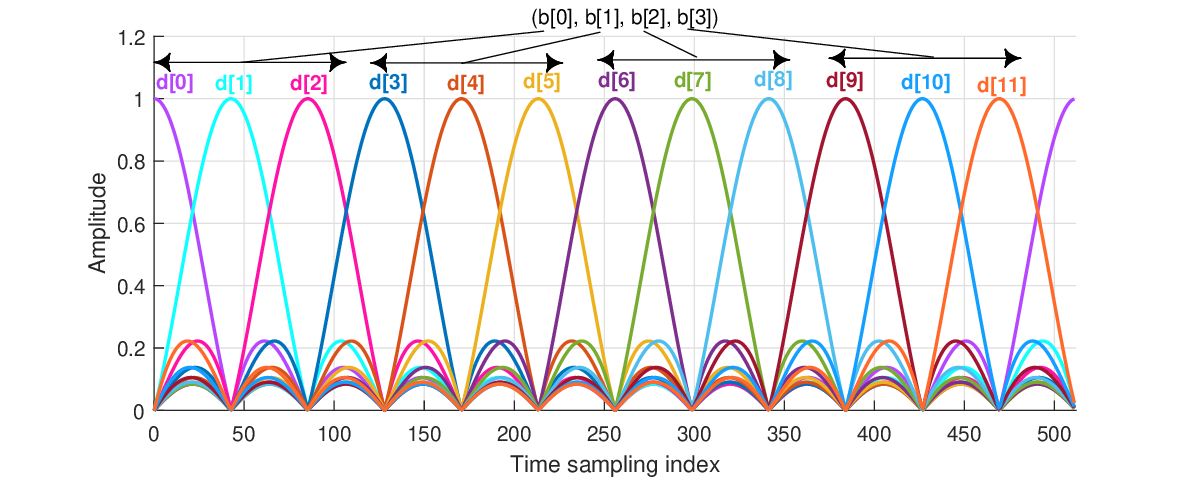}}
\subfigure[$\Nsymb= 9$,  with SE and FDSS ($T_{\rm shift}={ \scriptscriptstyle \frac{\Nfft}{2\Nsymb}}$ and $\beta=2$)]{\includegraphics[width=.49\textwidth]{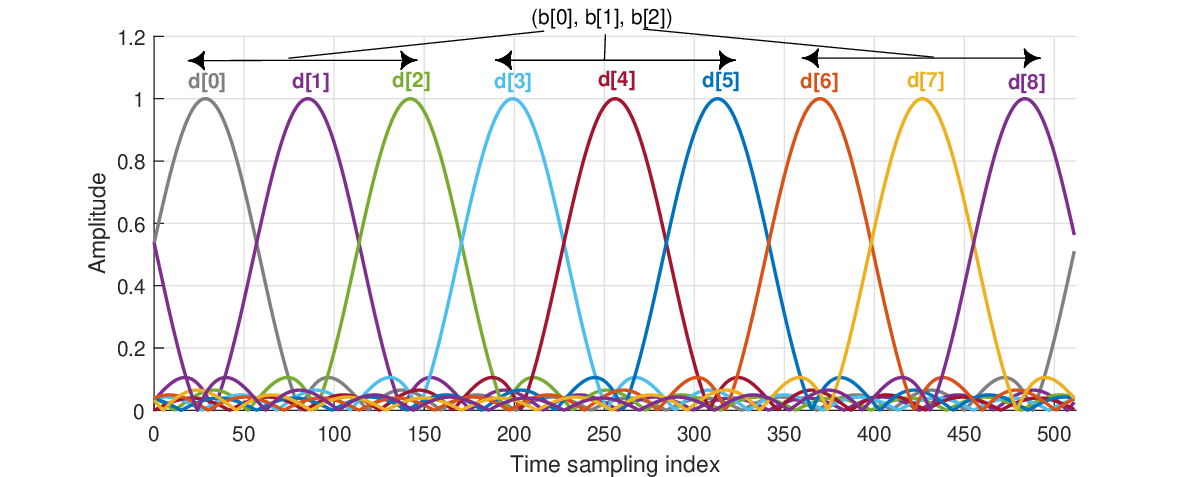}}
	\vspace{-0.4cm}
	\caption{Illustration of DFT-s-OFDM pulses with $\Nsc=12$ \label{fig:pulses}. a) 4 bits spread over $\Nsymb= 12$ orthogonal plain DFT-s-OFDM pulses.  b) 3 bits spread over $\Nsymb= 9\leq\Nsc$ non-orthogonal shaped  DFT-s-OFDM pulses .   }
 \vspace{-0.4cm}
\end{figure}

\subsection{Simplified Subcarrier Coefficients} 
Consider as bit-spreading a common overlaid sequence $r_0[n]$ possibly covered by linear phase ramp as 
\begin{equation} \label{eq:SingleBitSpreading}
r[m]=e^{\jrm \Phi m} r_0[m \; {\rm mod  }\; \Nseg ].
\end{equation}

\begin{figure*}[t] 
\centering
	 \vspace{-0.2cm}
\includegraphics[width=0.8\textwidth]{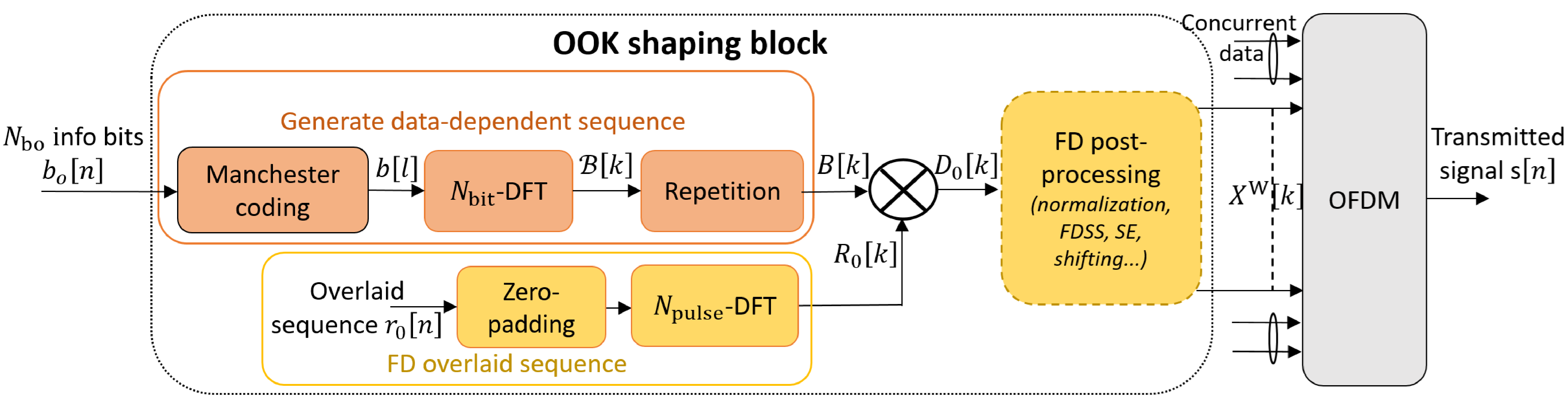} 
 \vspace{-0.1cm}
	\caption{Equivalent FD implementation of DFT-s-OFDM-based OOK in Fig.~\ref{fig:Tx} when bit-spreading is made of common overlaid sequence $\tilde{r}_l[n] =r_0[n]$.}
	\label{fig:Tx_FD}
	 \vspace{-0.3cm}
\end{figure*}

\subsubsection{Two-Layer FD Sequences} 
Then, the subcarrier coefficient computation can be simplified as layered FD sequences. 
\begin{Lem} \label{Lem:D=BRCont}
With bit-spreading~\eqref{eq:SingleBitSpreading},  
the pre-processed subcarrier coefficients \eqref{eq:D[k]} are given by
\begin{equation} \label{eq:D=BR}
D[k] = B\left(k-\frac{\Phi \Nsymb}{2\pi}\right) R_0\left(k-\frac{\Phi \Nsymb}{2\pi}\right)
\end{equation}
where  $B(f)  = \sum_{l=0}^{\Nbit-1} b[l] e^{-\jrm \frac{2 \pi }{\Nbit}f l}$  is the interpolated DFT\footnote{Analogous to the discrete-time Fourier transform (DTFT)} of the bits; and  $R_0(f) = \sum_{m=0}^{\Nseg-1}r_0[m]e^{-\jrm \frac{2\pi}{\Nsymb} fm}$ is the  
interpolated DFT of the overlaid sequence $r_0[n]$  
\end{Lem}
The proof is in Appendix~\ref{App:GeneralComputation}. 
Obviously, any cyclic shift $L$ in \eqref{eq:X^W[k]} can be equivalently implemented as a phase ramp with \mbox{$\Phi = -\frac{2\pi L}{\Nsymb}$}; but only some  phase ramp values can be implemented as a cyclic shift. On one hand, conventional DFT-s-OFDM is with $L=0$ and a phase ramp offers extra degrees of freedom for signal designs, but on the other hand 
cyclic shifting may be preferred for implementation  
specially when it corresponds to well-known operation such as e.g. the so-called `fftshift' with $L= -\Nsc/2$; and so we will keep both parameters.  
More specifically, we have

\begin{Cor} \label{Corr:D=BR} 
If $\frac{\Phi \Nsymb}{2\pi}$ is an integer, then $D[k]  =D_0[k-\frac{\Phi \Nsymb}{2\pi}]$ is a cyclic-shifted version of 
\begin{equation} \label{eq:D0=BR}
D_0[k] = B[k] R_0[k]
\end{equation}
where  $B[k]  =   \Bcal[k\; {\rm mod }\; \Nbit]$ is the cyclic extension of the $\Nbit$-point DFT of the bits 
\begin{equation}\label{eq:Bcal[k]} 
\Bcal[k]= \sum_{l=0}^{\Nbit-1} b[l] e^{-\jrm \frac{2 \pi }{\Nbit}k l}  
\end{equation} 
and where the FD overlaid sequence $R_0[k] = \sum_{m=0}^{\Nseg-1}r_0[m]e^{-\jrm \frac{2\pi}{\Nsymb}km}$ 
is the $\Nsymb$-point DFT of the overlaid sequence $r_0[n]$ with zero padding. 
\end{Cor}

Lem.~1 and Corr.~1  provides a different interpretation of DFT-s-OFDM-based OOK. 
The sequence $R_0[k]$, along with FDSS-SE, provides a common shaping window, constructing an  OOK symbol kernel, for all possible data-dependent sequences $B[k]$. This OOK kernel does not depends of the bits values but still depends on the bit number via $\Nseg$, setting a proper ON symbol duration. The sequence $B[k]$ carries the bit information by superposing multiple time-shifted version of the same OOK symbol kernel. Concretely, with~\eqref{eq:SingleBitSpreading} and if $\Nfft/\Nbit$ is an integer, the OOK symbols as given in~\eqref{eq:Ol[n]} become  a circular shift of the  OOK symbol kernel $O_0[n]$, up to  a phase coefficient, as 
\begin{equation}
O_l[n] = e^{-\jrm \frac{2 \pi l}{\Nbit}\left(L-\frac{\Phi \Nsymb}{2 \pi}\right)} O_0\left[n-l\frac{\Nfft}{\Nbit} \right]. 
\end{equation}
This is illustrated in Fig.~\ref{fig:TwoLayerFDseq} for the bit string  
$(0,1,0,1)$, with $\Nsc=72$, a ZC overlaid sequence, and no post-processing.   
As shown in the figure, the IFFT of $R_0[k]$ creates a single ON symbol while the IFFT of $B[k]$ generates sinc-pulses at different time instants, and the resulting OFDM signal is the convolution of these two signals, up to a scaling.

\begin{figure}[h] 
\centering
\vspace{-0.2cm}
\includegraphics[width=0.49\textwidth]{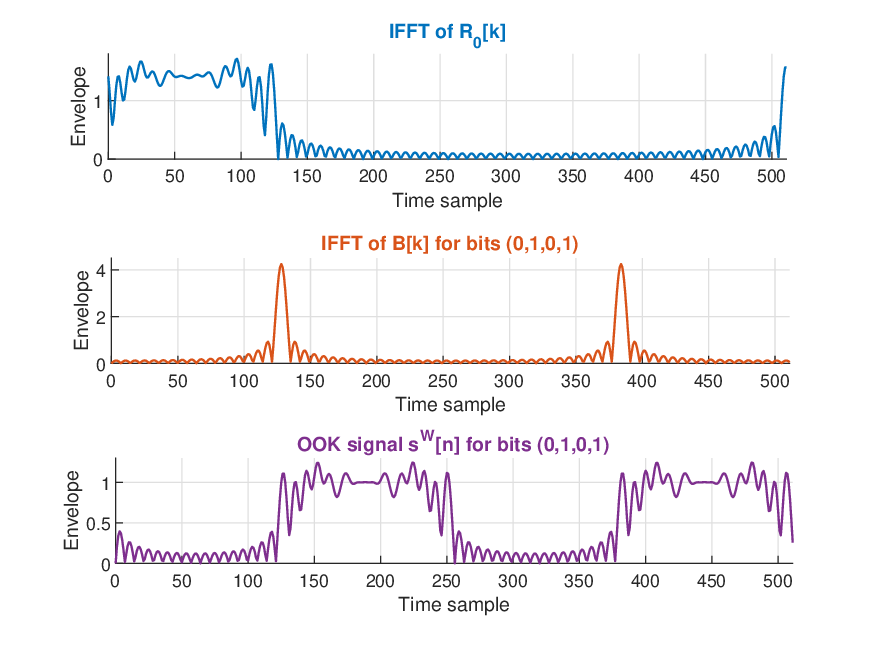}
\vspace{-1cm}
	\caption{Illustration of the layered FD sequences $B[k]$ and $R_0[k]$.}
	\label{fig:TwoLayerFDseq}
	 \vspace{-0.1cm}
\end{figure}

\subsubsection{Implementation Benefits} 
In the implementation of Fig.~\ref{fig:Tx}, the main complexity lies in the  data-dependent $\Nsc$-point DFT of Eq.~\eqref{eq:D[k]}. 
Corr.~\ref{Corr:D=BR} reveals a lower implementation complexity, as illustrated in Fig.~\ref{fig:Tx_FD}: As the sequence $R_0 [k]$ is independent of the bits, it can be computed offline and pre-stored. Also due to the repetitive structure of $B[k]$  and since typically $\Nbit \ll \Nsc$, computing the $\Nbit$-point DFT in Eq.~\eqref{eq:Bcal[k]} is less complex than computing the DFT in Eq.~\eqref{eq:D[k]}. 
Moreover, since the input consists of bits, the DFT in Eq.~\eqref{eq:D[k]}  does not actually require any multiplications. For small $\Nbit$, which is the regime of LP-WUS in NR, the computation of $\Bcal[k]$ can also be made offline and stored in columns of a $\Nbit \times 2^{\Nbit}$ look-up table, and $B[k]$ is then directly given by repetitions of $\Bcal[k]$ for any bandwidth size $\Nsc$.

\subsubsection{DFT of the Bits with Manchester Coding}
Furthermore, Manchester encoding can be absorbed in the computation of $\Bcal[k]$ to further reduce complexity and storage. 

\begin{Lem} \label{Lem::B[K]MC}  
With Manchester encoding, the DFT of the encoded-bits $\Bcal[k]$, for $k=0,\ldots,2\Nbo-1$, in Eq.~\eqref{eq:Bcal[k]}  can be directly computed from the info bits $ b_o[n]$ as
\begin{equation} \label{eq:B[K]MC}
\Bcal[k]= \sum_{n=0}^{\Nbo-1} e^{-\jrm \frac{\pi k}{\Nbo}(2n+ b_o[n])} . 
\end{equation}
\end{Lem}
The proof in Appendix~\ref{App:DFT_MCbits}.  In the case of a small $\Nbo$, the  values of $\Bcal[k]$ can be precomputed and stored in a $2\Nbo \times 2^{\Nbo}$ table such that each specific  information bit string is mapped to  a unique column of this table.    
Specifically, with $ \Nbo = 1 $ and information bit $b_o[0]$, then $B[k]$ reduces to $B[k]=(-1)^{kb_o[0]}$; 
with $ \Nbo = 2$ and information bits $(b_o[0],b_o[1])$, then $B[k]=(-\jrm)^{k b_o[0]}+(-\jrm)^{k(2+b_o[1])}$. So for $ \Nbo = 1 $ or $ \Nbo = 2 $, the sequences $B[k]$ are the repetition of the 2 or 4 coefficients of $\Bcal[k]$ in Table~\ref{tab:12bit}.

\begin{table}[t] \centering
	\vspace{-0.2cm}
	\caption{Coefficients $\Bcal[k]$ for Manchester-encoding of one or two information bits}
	\vspace{-0.2cm}
	\label{tab:12bit}
	\begin{tabular}{|c|c|c|}
	 \hline
	\multicolumn{3}{|c|}{$\Nbo = 1$} \\
 Info bit $b_o [0]$:        & 0 &1 \\ \hline
$\Bcal[0]$ & 1 & 1 \\ 
 $\Bcal[1]$ & 1 & -1 \\ 
 \hline
\end{tabular}
\hspace{1cm}
	\begin{tabular}{|c|c|c|c|c|}
	 \hline
	\multicolumn{5}{|c|}{$\Nbo = 2$} \\
 Info bit $(b_o [0],b_o [1])$ :     & (0,0) & (0,1)& (1,0)& (1,1) \\ \hline
$\Bcal[0]$ & $2$ & $2$ & $2$&$2$\\ 
 $\Bcal[1]$ & $0$ & $1+\jrm$ & $-1-\jrm$& $0$ \\ 
$\Bcal[2]$ & $2$& $0$ & $0$& $-2$\\ 
 $\Bcal[3]$ & $0$ & $1-\jrm$&  $-1+\jrm$& $0$ \\ 
 \hline
\end{tabular}
\vspace{-0.2cm}
\end{table}

\section{Rectangular-like OOK} 

The intuitive target for designing OOK signals is to create well-formed rectangular waveforms, as initially considered in the 3GPP  LP-WUS study~\cite{3GPPTR38.869} using either the LS-precoding method of~\cite{Mazloum20} or  a  DFT-s-OFDM-based implementation.

\subsection{Flattening ON Symbols in DFT-s-OFDM-based OOK} \label{sec:FlatOOK}
In order to limit envelope fluctuations within ON symbols, one should minimize the phase difference of two crossing neighboring pulses in~\eqref{eq:Ol[n]} to control their combining. This resembles low-PAPR  methods such as $\pi/2-$BPSK constellation for DFT-s-OFDM~\cite{KimTVT18}. Following this line of thinking, we consider a linear  phase ramp as bit-spreading sequence:    
$r[m]= e^{\jrm \Phi m}$. 
Then, as detailed in Appendix~\ref{App:Flat}, the coherent combining of two consecutive pulses can be minimized by selecting   
\begin{equation} \label{eq:PhiFlat1}
\Phi = \frac{\pi (2L+\Nsc-1)}{\Nsymb}.
\end{equation}

This result follows from the analysis of only two overlapping neighboring pulses, and therefore is only an approximation:
First, other neighboring pulses also contribute  to the envelope fluctuation. Second, if $\Nsymb\ll \Nsc$, the neighboring pulses may be so spread apart  
that their main lobes no longer overlap\footnote{In the case of no FDSS, the main lobe width of the kernel~\eqref{eq:Dirichlet}
is $2 \frac{\Nfft}{\Nsc}$ while pulses are shifted by $\frac{\Nfft}{\Nsymb}$ in~\eqref{eq:gm}; therefore if $\Nsymb\leq \Nsc/2$ the main lobes do not overlap anymore.}. 
When shaping is applied via FDSS, these discrepancies diminish as  the sidelobes are attenuated and the main lobes broaden.

In the case of conventional DFT-s-OFDM with 
 $\Nsymb=\Nsc$ and $L=0$, \eqref{eq:PhiFlat1} becomes $ \Phi  = \frac{\pi (\Nsc-1)}{\Nsc} 	\approx \pi$  and the bit-spreading sequence reduces to a $\pm 1$ alternation as $r[n]\approx (-1)^n$.   
Fig.~\ref{fig:PhaseRampImpact} shows the effect of the different  linear phases $\Phi= 0, \pi/4, \pi/2, \pi$, without and with FDSS. The waveform difference between~\eqref{eq:PhiFlat1} (exactly here $\Phi  =  0.97 \pi $)  and $\Phi= \pi$ is unnoticeable and thus not shown.  Using  $\Phi= 0$ results in the worst-case combining for consecutive pulses, leading to significant signal fluctuations and poor OOK signal properties. Without FDSS, a small phase ramp with $\Phi = \pi/4$ already substantially improves the signal shape in both the ON and OFF symbols.  
With FDSS, a larger phase ramp is needed to  further flatten the envelope of ON symbol, albeit at the cost of slower ON/OFF transitions.

\begin{figure}
\vspace{-0.2cm}
\subfigure[Without FDSS \label{fig:PhaseRampImpacta}]{\includegraphics[width=.49\textwidth]{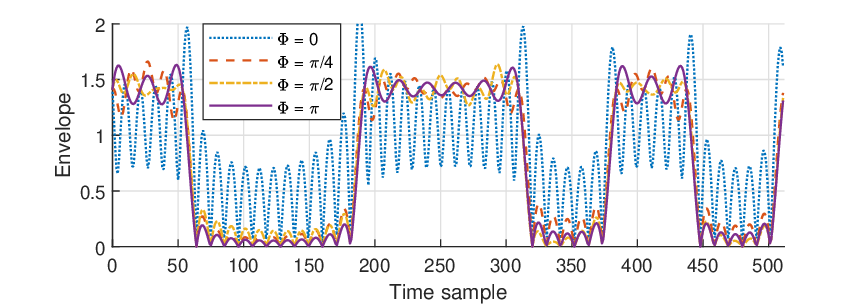}}
\subfigure[With FDSS  ($\beta =4$) \label{fig:PhaseRampImpactb}]{\includegraphics[width=.49\textwidth]{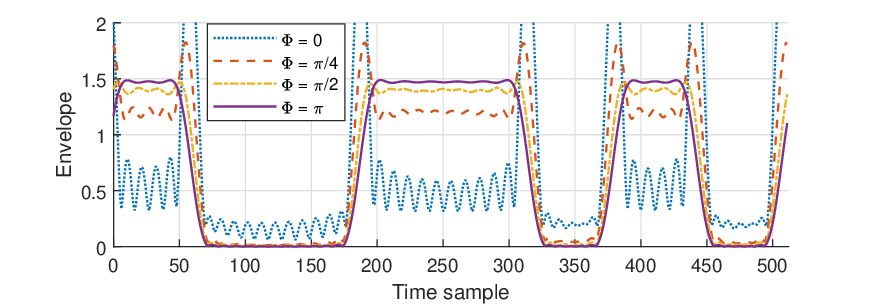}}
	\vspace{-0.4cm}
	\caption{OOK waveform flattening effect with increasing angle in linear phase ramp of bit-spreading sequence. Here 8-bit string $[1 0 0 1 1 0 1 0]$ is transmitted with $\Nsc=\Nsymb = 48$, and $L=0$.\label{fig:PhaseRampImpact}}
 \vspace{-0.4cm}
\end{figure}

With a phase ramp as bit-spreading, the WUS subcarrier coefficients~\eqref{eq:X^W[k]} can be further simplified as follows (see Appendix~\ref{App:GeneralComputation}). 
\begin{Cor} \label{Corr:XPhaseRamp}
With $r[m]= e^{\jrm \Phi m}$ for all $m$, the WUS subcarrier coefficients are given by 
\begin{equation} 
\!\! X^{W}[k] = \eta W[k] R_0\left(k-\textstyle \frac{\Phi\Nsymb }{2 \pi}+L\right) 
B\left(k-\textstyle \frac{\Phi\Nsymb }{2 \pi}+L\right)   
\end{equation}
where 
$R_0(0)= \Nsymb/\Nbit$ and otherwise 
\begin{equation} \label{eq:R0[k]RecOOK}
R_0(f) = \alpha e^{-\jrm\pi f (\frac1{\Nbit} -\frac1{\Nsymb})}
\frac{\sin \left(\frac{\pi}{\Nbit} f\right)}{\sin \left(\frac{\pi}{\Nsymb} f\right)}
\end{equation}
in which\footnote{The global phase $\alpha$ has no impact on the signal envelope and can be disregarded.} 
$  \alpha= e^{\jrm \pi (\frac{1}{\Nbit}- \frac{1}{\Nsymb})\left( \frac{\Phi \Nsymb}{2\pi}- L  \right) } $.
\end{Cor}

\subsection{Rectangular Zero-DC OOK}
The DC subcarrier is set to zero in the original IEEE 802.11  Wi-Fi signal format to avoid performance degradation due to DC-offset errors. 
Similarly, the DC subcarrier is nulled in 3GPP LTE standard but not necessarily in NR.  
A zero-DC OOK was first considered for IEEE 802.11ba WUS~\cite{LopezICC19,SahinGC18,Caballe19} 
but this constraint was eventually  relaxed to facilitate the generation of smoother ON/OFF symbols~\cite{LopezWifi18,Caballe19}. 

Using Lem.~\ref{Lem:D=BRCont}, we observe that this zero-DC constraint can be met if $R_0(k-\frac{\Phi \Nsymb}{2\pi})$ can be nulled at the middle subcarrier index. 
For rectangular OOK with $R_0(f)$ as given in~\eqref{eq:R0[k]RecOOK}, it is a sinc-like function whose null positions can be controlled by the angle of the phase ramp in the bit-spreading sequence. 
In general, a null at index $k_{\rm null}$ can be achieved by setting $\Phi$ satisfying $R_0 (k_{\rm null} -\frac{\Phi \Nsymb}{2\pi})=0$. With $R_0 (f)$ in~\eqref{eq:R0[k]RecOOK}, this is achieved by satisfying $\frac{\pi}{\Nbit} \left( k_{\rm null} -\frac{\Phi \Nsymb}{2 \pi} \right) =\lambda \pi$ with integer $\lambda \neq 0$, 
giving
\begin{equation}
\Phi = \frac{2 \pi}{\Nsymb} (k_{\rm null} - \lambda  \Nbit ).   
\end{equation}

\begin{figure}[t]
\vspace{-0.2cm}  
\subfigure[Example of particular signal envelope. \label{fig:nullDC_envelope}]{\includegraphics[width=.49\textwidth]{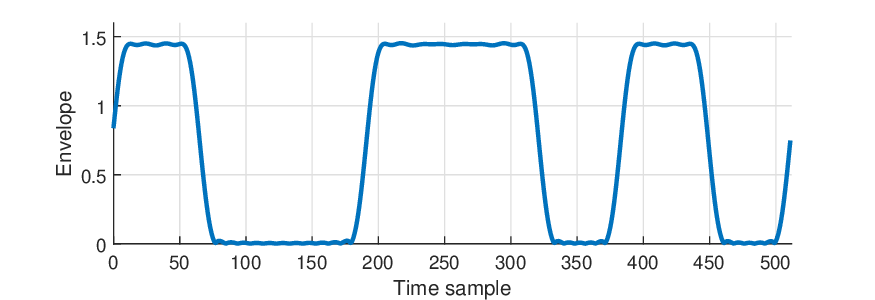}}
\subfigure[Average power of subcarrier coefficients for Manchester encoded 8-bit string  with alternating sign in the ramp phase. \label{fig:nullDC_spectrum}]{\includegraphics[width=.49\textwidth]{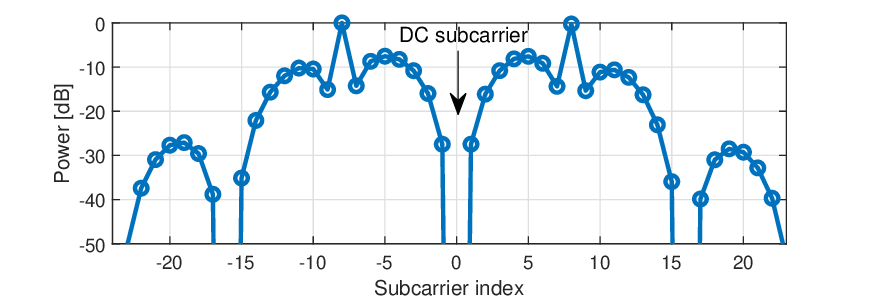}}
	\vspace{-0.4cm}
	\caption{Illustration of rectangular zero-DC OOK $\Nsc= 48$ using $\Phi=\pm \frac{2}{3}\pi$.  
	\label{fig:nullDC}}
	\vspace{-0.4cm}
	\end{figure}

Assuming $\Nsc=\Nsymb$ with DC subcarrier index $k_{\rm null}=\Nsc/2$,  
and taking $\lambda = 1$, the DC subcarrier can always be set to zero using $\Phi = \pi\left(1-\frac{2 \Nbit}{\Nsc} \right)$. If $\Nbit$ is small compared  to $\Nsc$, e.g., $\Nbit \leq \Nsc/4 $, then as discussed before, $\Phi$ would be large enough for  the OOK signal to retain a good rectangular shape. 

Due to the sinc-like function $R_0(f)$, the power distribution of the subcarriers coefficients will be concentrated  in a dominant lobe whose position is also shifted by the phase ramp. To ensure that the LP-WUS average spectrum is symmetrically distributed around the DC, different values of $\lambda$  can be selected, or the sign of $\Phi$ can be alternated in each OFDM symbol. 
Fig.~\ref{fig:nullDC} illustrates the resulting DC nulling effect with $\Nsymb= \Nsc= 48$, $\Nbit= 8 $, $k_{\rm null}=24$, and using $\Phi=\pm \frac{2}{3}\pi$ with alternating sign. As shown, the  middle subcarrier has been canceled, while the OOK waveform retains a rectangular shape since $\Phi = \frac{2}{3} \pi$ is close enough to $\pi$.

\subsection{Least-Squares (LS) Waveform}
In 3GPP LP-WUS and A-IoT studies,  
one of the main approaches was from~\cite{Mazloum20} which provides the LS approximation to an ideal rectangular OOK signal.  The corresponding subcarrier coefficients are obtained using a so-called ``LS precoding'' in~\cite{3GPPTR38.869}  
which is essentially a larger DFT of size matching the OFDM IFFT,  followed by a truncation to match the subcarrier allocation, see Appendix~\ref{App:LS} for further details.   
 
The LS  subcarrier coefficients obtained in~\cite{Mazloum20} can be written for $k=0,\ldots,\Nsc-1$  as
\begin{equation} \label{eq:XLS}
X^{W,\,{\rm LS}}[k] = \eta D^{\rm LS} \left[\left( k-\left\lfloor \frac{\Nsc}{2} \right\rfloor \right) {\rm mod} \, \Nfft \right]   
\end{equation}
where for $ k'=0,\ldots,\Nfft-1$
\begin{equation}  \label{eq:DFTLS}
D^{\rm LS}[k']=  \sum_{n=0}^{\Nfft-1} b_{\rm r}^{\rm LS} [n] e^{- j \frac{2\pi}{\Nfft}  n k'}  
\end{equation}
is a $\Nfft$-point DFT  of the spread-bit sequence $b_{\rm r}^{\rm LS}$ obtained by repeating the bit string $b[m]$ of length $\Nbit$ with spreading factor $\Nseg^{\rm LS}=\Nfft/\Nbit$ (assumed being an integer).  Since the perfectly-rectangular OOK signal $b_{\rm r}^{\rm LS}[k]$  is real, its spectrum is symmetrically concentrated around the zero frequency. Therefore the LS approximation to $b_{\rm r}^{\rm LS}[k]$ with constrained bandwidth %
is given by~\eqref{eq:XLS}  which preserves the dominant subcarrier coefficients around the zero frequency.  

The LS subcarrier coefficients~\eqref{eq:XLS} can be further simplified (see Appendix.~\ref{App:GeneralComputation}).  
\begin{Lem} \label{Prop:LS}
The LS coefficients \eqref{eq:XLS} are equivalently given by
\begin{equation} \label{eq:XLS2}
X^{W,\,{\rm LS}}[k] = \eta W^{{\rm LS}}[k] B\left[ k-\left\lfloor \frac{\Nsc}{2} \right\rfloor \right] 
\end{equation}
where  
$W^{\rm LS}[k]  = \frac{\Nfft}{\Nbit}$ for the middle subcarrier $k=\left\lfloor \frac{\Nsc}{2} \right\rfloor$, otherwise for other subcarrier indices $k\neq \left\lfloor \frac{\Nsc}{2} \right\rfloor$ 
\begin{equation}
\!\!\!\! W^{{\rm LS}}[k] = \alpha^{\rm LS} e^{-\jrm \pi k \left(\frac1{\Nbit} -\frac1{\Nfft}\right)}
\frac{\sin \left(\frac{\pi}{\Nbit} \left(\left\lfloor \frac{\Nsc}{2} \right\rfloor -k \right)\right)}{\sin \left(\frac{\pi}{\Nfft} \left(\left\lfloor \frac{\Nsc}{2} \right\rfloor -k \right)\right)}
\end{equation}
in which  $  \alpha^{\rm LS}= e^{\jrm \pi \left\lfloor \frac{\Nsc}{2} \right\rfloor \left(\frac{1}{\Nbit}- \frac{1}{\Nfft} \right)} $.  
\end{Lem}

By identification with Corr.~\ref{Corr:XPhaseRamp}, it follows 
 
\begin{Cor} \label{Cor:LSfromDFT}
The LS waveform from~\eqref{eq:XLS} can equivalently be generated by DFT-s-OFDM modulation as in~\eqref{eq:X^W[k]} with any $\Nsymb$ 
satisfying $\Nbit\leq \Nsymb \leq \Nsc$ and  $\Nsymb/\Nbit$ being an integer, if the bit-spreading sequence is a linear phase ramp satisfying $\Phi = \frac{2\pi\left(L+\left\lfloor \frac{\Nsc}{2} \right\rfloor \right)}{\Nsymb}$, and  the FDSS window is given for $k\neq \left\lfloor \frac{\Nsc}{2} \right\rfloor$ by   
\begin{equation} \label{eq:W[k]LS}
W[k] =  
e^{-\jrm \frac{2\pi k}{\Nfft} T_{\rm shift}}  
\frac{\alpha^{\rm LS}}{\alpha}
\frac{\sin \left(\frac{\pi}{\Nsymb}\left(\left\lfloor \frac{\Nsc}{2} \right\rfloor-k\right)\right)}{\sin \left(\frac{\pi}{\Nfft}\left(\left\lfloor \frac{\Nsc}{2} \right\rfloor-k\right)\right)}
\end{equation}
with  
$T_{\rm shift} = \frac{\Nfft-\Nsymb}{2\Nsymb}$, and $W[k]=\frac{\Nfft}{\Nsymb}$ for $k=\left\lfloor \frac{\Nsc}{2} \right\rfloor$.
\end{Cor}
 
It is worth noting that for  $\Nsc$ odd, $\left\lfloor \frac{\Nsc}{2} \right\rfloor =\frac{\Nsc-1}{2}$, and then $\Phi$ in Cor.~\ref{Cor:LSfromDFT} matches exactly~\eqref{eq:PhiFlat1}, which was derived to minimize consecutive pulse combining. For $\Nsc$ even, these values become approximately equal as $\Nsc$ increases.
   
For conventional DFT-s-OFDM with $L=0$ and $\Nsymb=\Nsc$,   $\Phi = \frac{2\pi}{\Nsc}\left\lfloor \frac{\Nsc}{2} \right\rfloor$ in Cor.~\ref{Cor:LSfromDFT},    
which further simplifies to $\Phi =\pi$ with  $\Nsc$ even. In this case, the LS waveform can be implemented using  a $\pm 1$ alternation bit-spreading sequence as  $r[n]=  (-1)^n$.  
 If one selects  $\Phi=0$ instead, then $L=-\left\lfloor \frac{\Nsc}{2} \right\rfloor$, corresponding to the  so-called ``fftshift'' operation that circularly shift the $0$th subcarrier coefficient to the middle subcarrier, but applies after repetition if $\Nsymb < \Nsc$ and not directly after DFT precoding. 
Intuitively, the LS waveform is obtained by mapping the DC component of the input bit string at the input of DFT precoder to the DC component of the output OFDM signal. An effect of this DC-to-DC mapping is that the LP-WUS signal is a properly interpolated version of  the input bits. 

 \begin{figure}[t] 
\centering
\vspace{-0.0cm}
\includegraphics[width=0.5\textwidth]{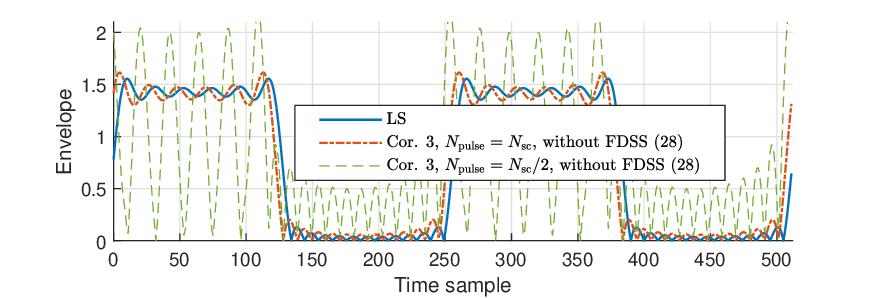}
\vspace{-0.6cm}
	\caption{Illustration of the minor effect of the FDSS window~\eqref{eq:W[k]LS} of  Cor.~\ref{Cor:LSfromDFT}.}
	\label{fig:LSvsDFT}
	\vspace{-0.1cm}
\end{figure}

 \begin{figure}[t] 
\centering
\vspace{-0.0cm}
\includegraphics[width=0.49\textwidth]{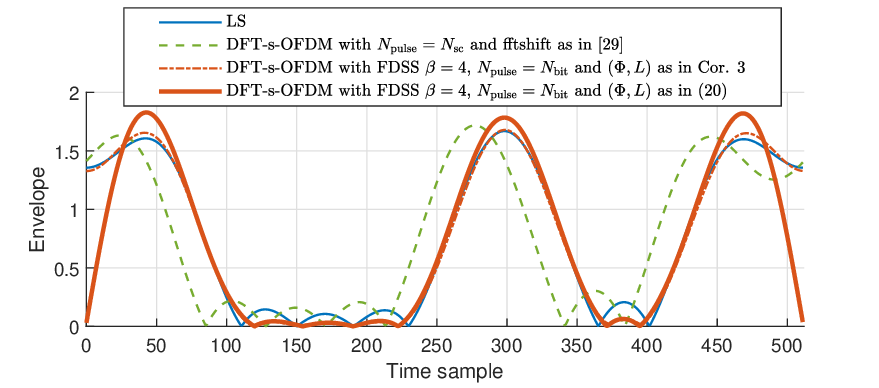}
\vspace{-0.5cm}
	\caption{Example of $\Nbit=6$ bits OOK ($[1 0 0 1 0 1]$) sent in  narrow band of $\Nsc = 12$ as considered for Ambient IoT application~\cite{R1-2403954}.}
	\label{fig:AIoT_OOK}
	\vspace{-0.2cm}
	\end{figure}

The $(\Phi,L)$ relationship  has an important impact on the signal envelope, and departing too much from it can result in large envelope fluctuation as in Fig.~\ref{fig:PhaseRampImpact}.  Conversely,   the effect of the FDSS window~\eqref{eq:W[k]LS} can be negligible especially for $\Nsymb = \Nsc$ large: When $\Nsymb = \Nsc$ the amplitude of~\eqref{eq:W[k]LS} is nearly constant (almost numerically identical to a Kaiser window with $\beta\approx 2$), and its phase  becomes  negligible as $\Nsymb$ grows. 
Therefore, choosing the largest possible DFT precoder with $\Nsymb=\Nsc$ may allow the FDSS to be removed  with little impact. However when $\Nsymb < \Nsc $ the FDSS effect becomes more significant.   This is illustrated in Fig.~\ref{fig:LSvsDFT}, which compares the LS waveform obtained from   Cor.~\ref{Cor:LSfromDFT} with and without the FDSS window~\eqref{eq:W[k]LS}, using $\Nbit=4$ and $\Nsc=\Nsymb=48$.

Finally, we remark that in narrowband scenarios such as those in A-IoT applications where larger bit rate and smaller bandwidth than LP-WUS are considered, approximating an ideal --inherently wideband-- rectangular signal becomes difficult. In such cases, investigating alternative DFT-s-OFDM-based designs may offer interesting solutions. 
We provide an example in Fig.~\ref{fig:AIoT_OOK} for $\Nbit= 6$ bits per OFDM symbol using $\Nsc=12$ subcarriers, as considered in~\cite{R1-2403954}. This figure illustrates the DFT-s-OFDM-based implementation therein,  corresponding here to $\Nsymb= \Nsc$, $L= \Nsc/2$ (referred as `fftshift'), and $\Phi =0$. This implementation satisfies the relationship between $L$ and $\Phi$ as in Cor.~\ref{Cor:LSfromDFT} but differs from the LS waveform by not including the FDSS~\eqref{eq:W[k]LS}, which results primarily in a time shifting difference. 
Fig.~\ref{fig:AIoT_OOK} also shows an alternative signal design using the minimum number of DFT-s-OFDM pulses $\Nsymb=\Nbit$ combined with more spectrum shaping using a Kaiser window with $\beta=4$; and both the $(\Phi,L)$ relationship from Cor.~\ref{Cor:LSfromDFT}  and from~\eqref{eq:PhiFlat1}. Since $\Nsc$ is even and small, we have $\Phi= \pi \left(L/3+ 2 \right)$ in Cor.~\ref{Cor:LSfromDFT}, which differs from  $\Phi=\pi \left(L/3+11/6\right)$ in~\eqref{eq:PhiFlat1}. As shown, the DFT-s-OFDM waveform with $\Nsymb=\Nbit$ and $\Phi$ from~\eqref{eq:PhiFlat1} provides better pulse transition and smaller sidelobe in OFF symbols, resulting in power-boosted ON symbols which could facilitate detection.

\begin{figure*}
\vspace{-0.2cm}  
\centering
\subfigure[$\Nsymb=\Nsc/2$, envelope. ]{\includegraphics[width=.49\textwidth]{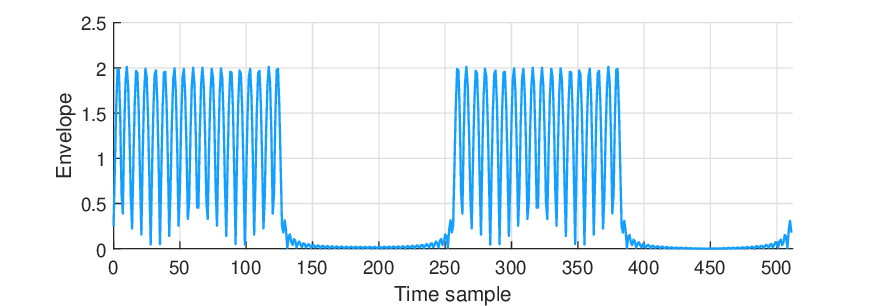}}
\subfigure[$\Nsymb=\Nsc/2$, power of subcarrier coefficients. ]{\includegraphics[width=.49\textwidth]{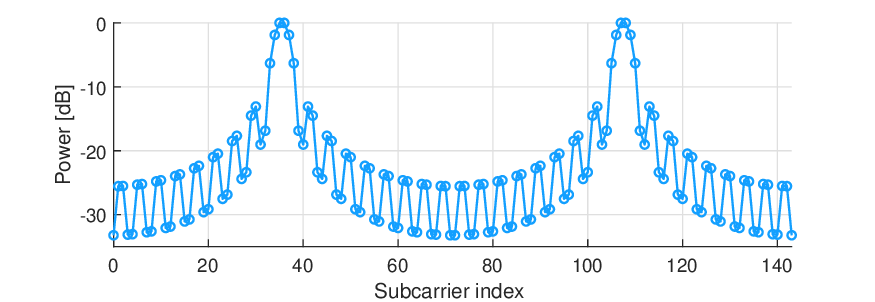}}
\subfigure[$\Nsymb=\Nsc/3$, envelope.  ]{\includegraphics[width=.49\textwidth]{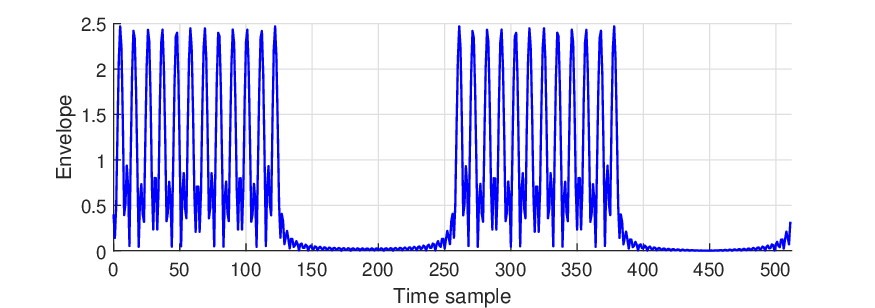}}
\subfigure[$\Nsymb=\Nsc/3$, power of subcarrier coefficients.  ]{\includegraphics[width=.49\textwidth]{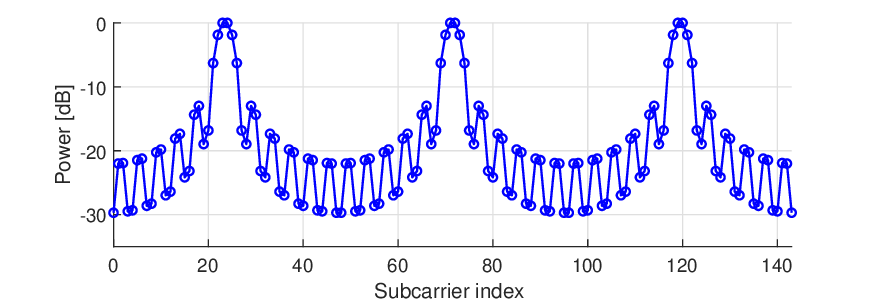}}
	\vspace{-0.2cm}
	\caption{Illustration of rectangular-like OOK with repeated spectrum using $ \Phi = \frac{\pi (2L+\Nsymb-1)}{\Nsymb}$; with $\Nsc=144$ for bit string $[1 0 1 0]$.}
	\label{fig:FreqRep}
	\vspace{-0.4cm}
	\end{figure*}

\vspace{-0.0cm}
\section{Shaped OOK for 5G  LP-WUS}   
In 5G LP-WUS,  only $\Nbit= 1$, $2$, or $4$  can be transmitted, while the ADC sampling rate at the LR is assumed to be  much higher than  such low OOK data rates. Indeed, it was shown in~\cite{3GPPTR38.869} that an ADC sampling rate of, for example, 3.84 MHz was usable to improve performance while having little impact on the LR's power consumption. With sufficient ADC precision at LR, the OOK detection performance principally depends on  the energy split between the ON and OFF symbols, while minimizing the envelope fluctuation inside ON symbols may actually matter little. Therefore, the standard rectangular-like OOK waveform as targeted in the previous section does not provide the best performance for LP-WUS, and in this section, we propose and analyze FD and TD shaping features  
 for performance improvement in such application environments. 
 
\subsection{FD Shaping for Improved Performance Against Frequency-Selective Fading} 
An LP-WUS bandwidth no larger than $5$ MHz was recommended in~\cite{3GPPTR38.869} to facilitate deployment, amounting to less than 144 subcarriers with $30$  kHz spacing in 3GPP numerology. 
At the same time, larger signal bandwidth can increase  frequency diversity and thus robustness against frequency selective fading,
such that  most sources in~\cite{3GPPTR38.869} considered exploiting as much  as possible of this $5$ MHz. In this context of a relatively large bandwidth with a low data rate, rectangular-like OOK signals discussed in the previous section have concentrated power spectrum that makes them vulnerable to frequency-selectivity of wireless fading.  Alternative signal designs with more spread-out spectrum, but without compromising much on the ON/OFF energy split, can thus improve performance\footnote{However, for the purpose of designing OOK synchronization signal, an OOK signal with narrow spectrum  can provide benefits for performing frequency-synchronization in frequency-selective channels~\cite{R1-2403948}.}.  Below are  methods to achieve this.

\subsubsection{Frequency-Repetition} \label{sec:FreqRep}
Starting  from the rectangular-like OOK design discussed so far,  
a generic method to achieve higher frequency diversity  is to repeat its subcarrier coefficients on a wider bandwidth.  
This can be achieved in~\eqref{eq:X^W[k]} by selecting $\Nsymb \ll \Nsc$, i.e. modulating less, and more spread-apart, DFT-s-OFDM pulses, leading inevitably to a comb-like signal shape. In order to guarantee a good ON/OFF energy split,  we first apply~\eqref{eq:PhiFlat1} to the intermediate narrow band of $\Nsymb$ subcarriers, i.e. setting $ \Phi = \frac{\pi (2L+\Nsymb-1)}{\Nsymb}$, before repeating in~\eqref{eq:X^W[k]} to populate the $\Nsc$ subcarriers. An example of the obtained waveforms and power distribution over subcarrier indices is shown on Fig.~\ref{fig:FreqRep} with $\Nsc=144$ and $\Nbit = 4$. No FDSS is used here other than a linear phase for time-shifting the signal by half a pulse.  
As it can be seen the ON symbols cannot be flattened since the main lobes of neighboring pulses simply do not overlap, but by selecting a proper $(\Phi,L)$,  the OFF symbols can be guaranteed to be close to zero. Similarly to Fig.~\ref{fig:PhaseRampImpact}, variations over $(\Phi,L)$ can, in fact, be tolerated with little effect on the signal shape until high envelope lobes are created in the OFF symbols, similar to the green curve of Fig.~\ref{fig:LSvsDFT}. 
The drawback of this approach is as $\Nsymb \to \Nbit$, the signal becomes less robust to lower sampling rate if considered, and the PAPR of the overall OFDM transmission may increase even beyond what it is already with QAM symbols. 
 
\begin{figure*}
\vspace{-0.2cm}  
\centering
\subfigure[ Envelope with root $u=1$ and shift $s=0$. \label{fig:ZCa}]{\includegraphics[width=.49\textwidth]{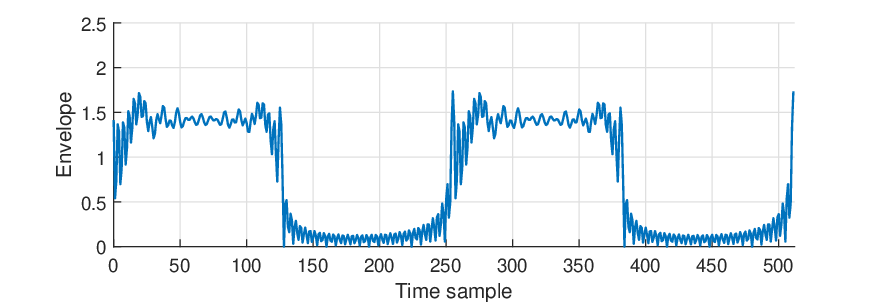}}
\subfigure[Envelope with root   $u=(N_{\rm seg}/2+1)$ and shift $s=0$. \label{fig:ZCb}]{\includegraphics[width=.49\textwidth]{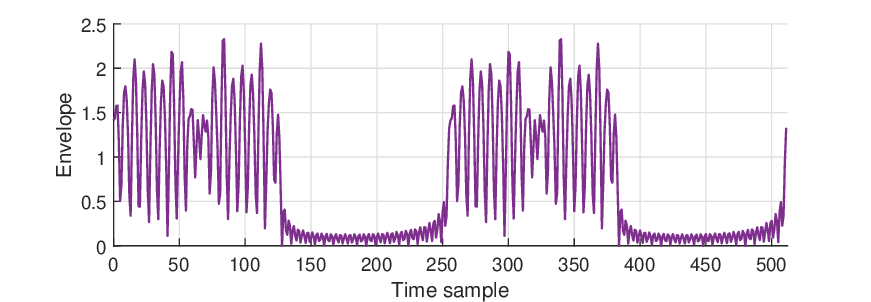}}
\subfigure[Envelope with root $u=(N_{\rm seg}-1)$ and shift $s=0$.\label{fig:ZCc}]{\includegraphics[width=.49\textwidth]{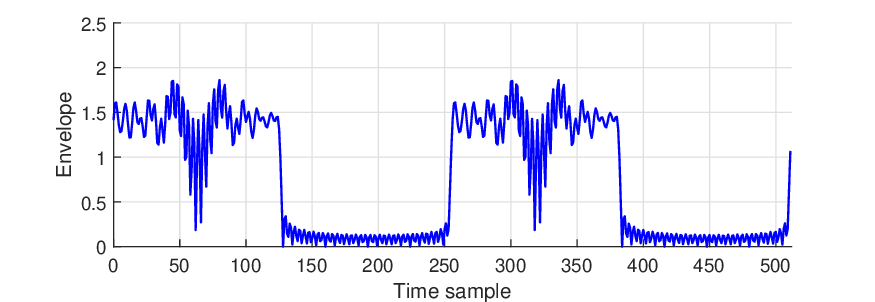}}
\subfigure[  Envelope with root  $u=(N_{\rm seg}-1)$ and shift $s=N_{\rm seg}/2$.  \label{fig:ZCd}]{\includegraphics[width=.49\textwidth]{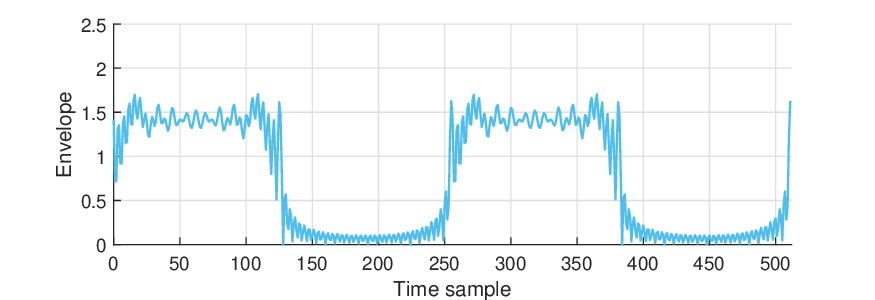}}
\subfigure[Deviation of ZC phase increment from $\Phi \approx \pi$ in~\eqref{eq:PhiFlat1}  for signals in (a)-(d)  \label{fig:ZCf}]{\includegraphics[width=.49\textwidth]{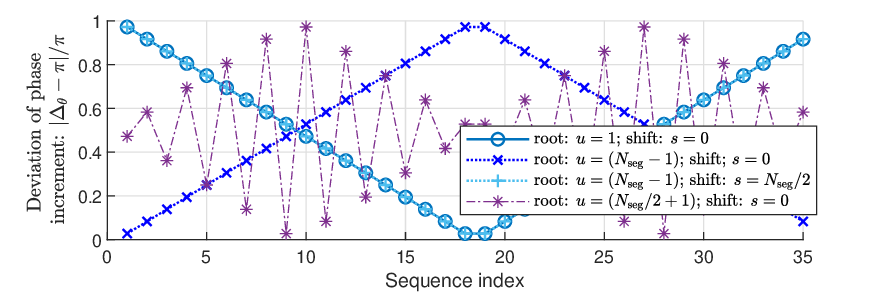}}
\subfigure[Power of subcarrier coefficients for signals in (a)-(d) \label{fig:ZCe}]{\includegraphics[width=.49\textwidth]{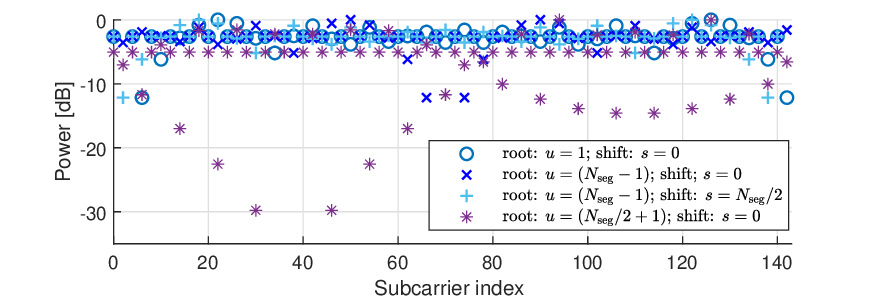}}
	\vspace{-0.4cm}
	\caption{Examples of OOK signal where $r_0[m]$ is a ZC sequence; $\Nsc=144$ with bit string $[1 0 1 0]$ and sequence length is $\Nseg=36$.}
	\label{fig:ZC}
	\vspace{-0.2cm}
	\end{figure*}

\subsubsection{Phase Scrambling in Overlaid Sequence} 
Methods for varying the phase of the overlaid sequence  $r_0[k]$ in order to flatten the spectrum were considered in~\cite{R1-2306234}. While by scrambling the phase of the spreading sequence, the envelope shape of the ON symbol may become less controlled, the ON/OFF energy split can still be expected to remain good as it is unlikely that many distant sidelobes in the OFF symbols would coherently combine (cf. $O_l[n]$ in~\eqref{eq:Ol[n]} for samples outside of $[l\Nseg, \, (l+1)\Nseg-1] $).

As shown from Corr.~\ref{Corr:D=BR}, the FD power distribution  depends on two sequences: $R_0[k]$ from the overlaid sequence, and $B[k]$ from the bits 
which inherently has  
a repetition structure  such that its power is already well-distributed across the band. The FD power distribution of DFT-s-OFDM-based OOK satisfies then the following property shown in Appendix~\ref{App:ProofLemPSD}.
\begin{Lem}  \label{Lem:E[|D_0[k]|^2]} 
With Manchester coding\footnote{Without Manchester coding, %
$ E[|B[k]|^2]= \frac{\Nbit(\Nbit+1)}{4}  $ for $k= (0 \, {\rm mod}  \,  \Nbit)$ and  
   $ E[|B[k]|^2] = \frac{\Nbit}{4}$ otherwise; and more than half, precisely $\frac{(\Nbit+1)}{2\Nbit}$, of the \emph{average} sum power is allocated on $\Ccal$.}, 
	half of the sum power of $\{D_0[k]\}_{k=0}^{\Nsymb-1}$  
	is always allocated on a regular comb of subcarrier indices $\Ccal = \{k'\Nbit\}_{k'=0}^{\Nseg-1}$. On average, the power distribution is given by $ E[|D_0[k]|^2] =  |R_0[k]|^2 E[|B[k]|^2]$ with  
\begin{equation} \label{eq:E[|B[k]|^2]} 
E[|B[k]|^2] = 
\begin{cases}    \frac{\Nbit^2}{4}  \quad\quad   {\rm for} \;\; ( k \,{ \rm mod}\,\Nbit) = 0  \\
    \frac{\Nbit}{4} \left(1-\cos \frac{2\pi k }{\Nbit}\right) \quad {\rm otherwise}
		\end{cases}.  
\end{equation}   
\end{Lem}
 
While having half of the power allocation on this regular comb of subcarriers $\Ccal$ is independent of $R_0[k]$, the power level fluctuation on $\Ccal$ depends on  it. Meanwhile, the power level of $R_0[k]$  on $\Ccal$ 
is directly given by the DFT of $r_0[l]$, $ R_0'[k] = \sum_{l=0}^{\Nseg-1}r_0[l] e^{-\jrm \frac{2 \pi}{\Nseg}  kl} $, since $R_0[k]$ is an  $\Nbit$-upsampled and sinc-interpolated version of $R_0'[k]$ given by 
\begin{equation} \label{eq:R_0[k]Interpolation}
\!\!\!R_0[k] = \begin{cases} 
R'_0[k/\Nbit ] \quad \quad\quad  {\rm for} \;\; ( k \,{ \rm mod}\,\Nbit) = 0 \\
 \displaystyle  \frac{1}{\Nseg}\!\!\! \sum_{h=0}^{\Nseg-1}\!\!\! R'_0[h] I[k-\Nbit h] \quad {\rm otherwise,} \!\!\!\end{cases}
\end{equation}
with \mbox{$I[k] =e^{-\jrm  \frac{\pi (\Nseg-1)k }{\Nsymb }} \frac{\sin\left( \frac{\pi k}{\Nbit} \right)}{\sin\left( \frac{\pi k }{\Nsymb}\right)}$}, see Appendix~\ref{App:R0[k]}.
 
So by combining Lem.~\ref{Lem:E[|D_0[k]|^2]} with~\eqref{eq:R_0[k]Interpolation}, the average power of the subcarrier coefficients  $E[|D_0[k]|^2]$ is completely characterized; from which one sees that  half of the average sum power 
can be directly controlled by  the overlaid sequence $r_0[l]$ to be distributed with a constant level over the band.

\paragraph{Random Overlaid Sequence}
Randomizing the phase of $r_0[m]$  has been proposed in~\cite{R1-2306234} to ensure a  well-spread energy along the band. Indeed, this approach equalizes on average the power of its DFT coefficients since it leads to  $E[|R_0[k]|^2]= \Nseg,\, \forall k$ , which combined with Lem.~\ref{Lem:E[|D_0[k]|^2]} gives 
\begin{Cor} \label{Corr:D[k]withRandSeq}
 With random overlaid sequence such that $E[r_0[m]r_0^*[m']]= \delta_{m,m'}$,  half of the average sum power of $\{D_0[k]\}_{k=0}^{\Nsymb-1}$  is equally distributed on a $\Nbit$-spaced comb. 
\end{Cor} 
The drawback of this approach is that the signal is not controlled anymore, and therefore, its knowledge could not be exploited by receivers with higher-end capabilities.

\paragraph{Overlaid CAZAC sequence}
The approach retained for 5G LP-WUS is to use a ZC sequence\footnote{For the purpose also of carrying  the bits by overlaid sequence selection, in parallel to OOK, enabling improved performance for higher-end LRs.}  
which is a constant-envelope zero autocorrelation (CAZAC) sequence.  A well-known property of a CAZAC sequence is that its DFT is also of constant magnitude~\cite{PopovicArxiv20}, such that $|R_0'[h]|^2=\Nseg,\, \forall h$ in~\eqref{eq:R_0[k]Interpolation}.  Combined with Lem.~\ref{Lem:E[|D_0[k]|^2]} this gives 
 the same result  as in Corr.~\ref{Corr:D[k]withRandSeq} but for each OFDM symbol without averaging.  
\begin{Cor} \label{Corr:D[k]withZC}
With a CAZAC overlaid sequence,   half of the sum power of $\{D_0[k]\}_{k=0}^{\Nsymb-1}$  is equally distributed on a regular $\Nbit$-spaced comb. 
\end{Cor}
 
\emph{Envelope Shape With Overlaid ZC Sequences:} More precisely, if the overlaid sequence is a ZC sequence defined as $r_0[m] = e^{-\jrm \frac{\pi u(m+s)(m+s+\delta)}{\Nseg}}$ where $\delta  =(\Nseg \mod 2)$, different root $1 \leq u\leq \Nseg-1 $ ($u$ must be relative prime to $\Nseg$) and  shift  $0 \leq s  \leq \Nseg-1 $ would create different ON symbol shapes as shown in Figs.~\ref{fig:ZC}(a)-(d) in which $\Nsymb =\Nsc=144$, $\Phi=L=0$, $\Nbit =4$ and $\Nseg=36$.  
To apprehend this, remark that the phase increment between consecutive elements in a ZC sequence is a linear function as
\begin{equation}
\Delta_{\theta} [m] = \angle \frac{r_0[m]}{r_0[m-1]} = -\pi \frac{u(2m+2s-1+\delta)}{\Nseg} \mod 2 \pi.
\end{equation} 
Comparing to  
\eqref{eq:PhiFlat1} for flattening ON symbols, the closer $\Delta_{\theta} [m]$ would be constant and close to $\Phi \approx \pi$ in this case, the better the DFT-s-OFDM pulses  are expected to constructively combine. The deviations  of this linear phase increment  $\Delta_{\theta}$ from a constant increment of $\pi$ are shown on Fig.~\ref{fig:ZCf}. With roots closer to $(0 \mod \Nseg)$, e.g. $u =1$ or $u=(\Nseg-1)$, the increment changes slower than with `middle roots', e.g. $u = (\Nseg/2 +1)$, which is better for pulse combining. The largest deviations for $(u =1,s=0)$ are at the beginning and at the end of the sequence which affect the ON symbol edges. On the contrary, with  $(u =\Nseg-1,s=0)$, the largest deviation is in the middle of the sequence, setting the highest envelope fluctuation  right in the middle of the ON symbol. By selecting a different shift $s=\Nseg/2$ for $u =(\Nseg-1)$, the deviation of the phase increment then overlaps with case $(u =1,s=0)$, and thus both ZC sequences create  almost the same OOK envelope shape. Finally, the corresponding powers of subcarrier coefficients for these signals are shown in Fig.~\ref{fig:ZCe}. Every two coefficients are equal to zero which is due to $B[k]$ being equal here to the repetition of the first column in the second table of Table~\ref{tab:12bit}.  Every fourth coefficient has the same power as shown in Corr.~\ref{Corr:D[k]withZC}, accounting for half of the sum power as stated in Lem.~\ref{Lem:E[|D_0[k]|^2]}. Other coefficients' power fluctuates according to~\eqref{eq:R_0[k]Interpolation}.

\subsection{TD Shaping for Improved Robustness Against Timing Error}  
It was observed in~\cite{3GPPTR38.869} that while OOK-modulated LP-WUS is robust against frequency synchronization error, it is sensitive to timing error. The larger $\Nbit$ is, the smaller the OOK symbols are, and the less maximum timing error can be tolerated. 
LRs are expected to use low-cost oscillator with low-precision frequency synchronization, resulting in timing clock drifting. Since LP-WUS is a sporadic message based on needs, its precise transmission time and thus resulting timing drift from previous synchronization acquisition cannot be known at LR.  We will model this timing error as uniformly distributed in range $[-\tau_{\rm err}, \, \tau_{\rm err}] $ where  $\tau_{\rm err}$ is the maximum expected timing error around current receiver timing, given for example from the synchronization signal periodicity.  
 
On Fig.~\ref{fig:ConcentratedOOK}, window boundaries for OOK symbol detection are depicted by thick black lines.   A normal OOK, where ON symbols have been designed to have their energy spread over one symbol duration as shown in Fig.~\ref{fig:ConcentratedOOK}(a), will lead, in case of timing offset at the LR, to a large energy leakage to the adjacent symbol detection window  as shown on Fig.~\ref{fig:ConcentratedOOK}(b). However,  by concentrating the energy of the ON symbols from the right side toward its center as in Fig.~\ref{fig:ConcentratedOOK}(c), this energy leakage at the receiver  can be  significantly reduced as shown on Fig.~\ref{fig:ConcentratedOOK}(d), bringing robustness against positive timing offset. Equivalently, concentrating the ON symbol from the left side would bring robustness against negative timing offset. 

Similarly, channel multi-path delays are also likely to cause inter-OOK-symbol interference within one OFDM symbol if there is no guard time between them. 
Compared to drifting from oscillator which may be assumed to have a symmetrical  effect, channel time-dispersion is more likely to result in replica signals with positive delays w.r.t. synchronization timing.

\begin{figure}[t]
\vspace{-0.0cm}  
\centering
\subfigure[Normal OOK, transmitted.]{\includegraphics[width=.49\textwidth]{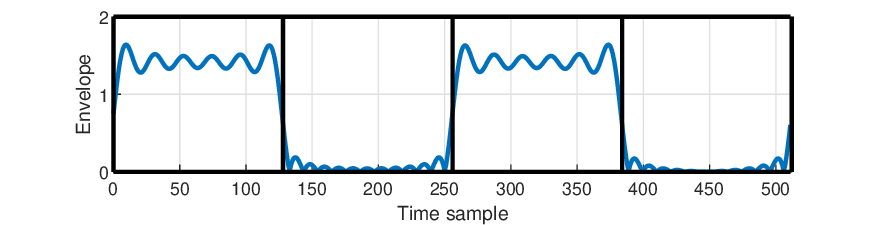}}
\subfigure[Normal OOK, received. ]{\includegraphics[width=.49\textwidth]{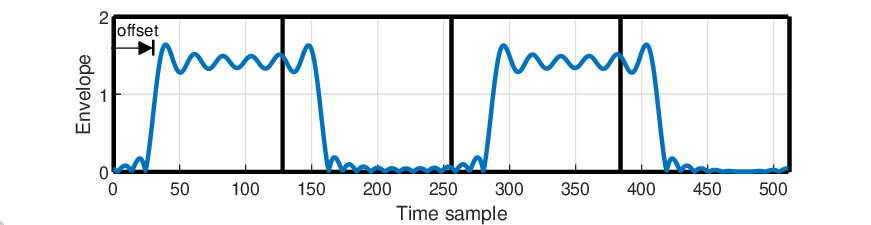}}
\subfigure[Concentrated OOK, transmitted.]{\includegraphics[width=.49\textwidth]{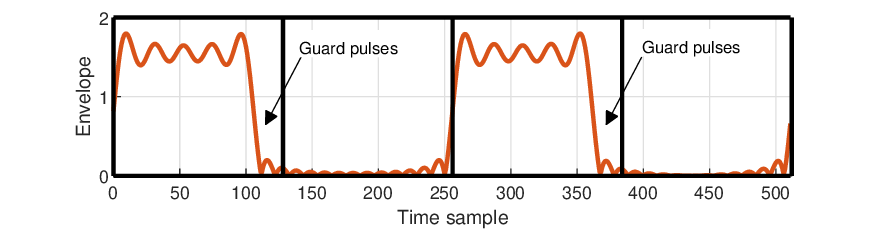}}
\subfigure[Concentrated OOK, received.]{\includegraphics[width=.49\textwidth]{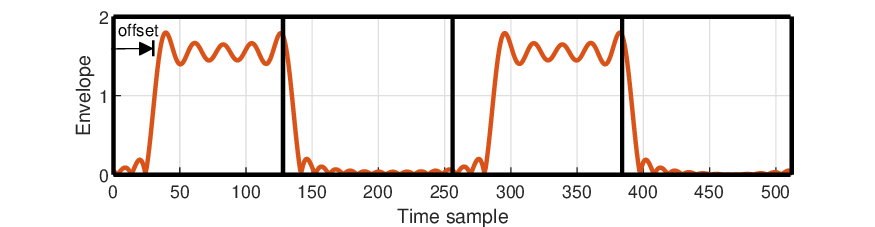}}
	\vspace{-0.4cm}
	\caption{Illustration of concentrating ON symbol energy against timing error.}
	\label{fig:ConcentratedOOK}
	\vspace{-0.4cm}
	\end{figure}

\subsubsection{Concentrated DFT-s-OFDM based OOK} For DFT-s-OFDM based OOK, this energy concentration can be achieved through windowing inside the overlaid sequences $r_0[m]$. To minimize any energy leakage, we consider rectangular windowing by setting the first and last elements in the overlaid sequence $r_0[m]$ to zero. This corresponds to introduce some guard --unmodulated-- DFT-s-OFDM pulses in~\eqref{eq:Ol[n]}. Using  $N_{\rm Lgp}$ left and $N_{\rm Rgp}$ right guard pulses, the overlaid sequence satisfies now
\begin{equation}
r_0[m] = 0 \text{ for }  \begin{cases}  0\leq m \leq N_{\rm Lgp}-1 \\  \Nseg-N_{\rm Rgp}\leq m \leq \Nseg-1 \end{cases}.
\end{equation}
Obviously, $N_{\rm Lgp} +N_{\rm Rgp} < \Nseg$ must be fulfilled, and the resulting effective length for using, e.g., a ZC overlaid sequence becomes $\Nseg' =\Nseg -N_{\rm Lgp}- N_{\rm Rgp}$.

\subsubsection{How Many Guard Pulses?} 
To mitigate also channel time dispersion, it may be beneficial to consider more guard pulses on the right than on the left $N_{\rm Rgp}\geq N_{\rm Lgp}$, and thus an asymmetrical ON symbol shape.  
As one OFDM symbol has a time duration of $1/f_{\rm sc} $,   
one pulse has most of its energy  spanning a duration of $1/(\Nsymb f_{\rm sc}) $. So, in order to cushion a maximum timing error of $\tau_{\rm err}$ in addition to an expected channel spreading of $\tau_h$, a rule of thumb for dimensioning the number of guard pulses is therefore 
\begin{eqnarray}
N_{\rm Lgp} &\approx& \tau_{\rm err} \Nsymb  f_{\rm sc} \label{eq:NLgp}\\
N_{\rm Rgp} &\approx&  (\tau_{\rm err}+ \tau_h) \Nsymb  f_{\rm sc}. \label{eq:NRgp}
\end{eqnarray}    

\subsubsection{Concentrated Receiver Window}
It is well known from radar literature that a low pulse duty cycle increase peak power, and thus the SNR for a matching receiver integrating then the same signal energy with less noise.    
This method was  proposed in~\cite{Lopez_PartialOOK} for improving Wi-Fi WUS.  
However, as pointed in~\cite{Lopez_PartialOOK}, by matching the detection window to the concentrated OOK, the SNR would increase at the cost of an increased sensitivity to timing error. 
Alternatively, one could also consider to only concentrate the detection window which would in turn improve the robustness against timing error at the cost of a decreased SNR. Hence, for both cases,  using a concentrated detection window leads to a trade-off between timing error robustness and SNR. 
Differently, if only concentrating  the transmitted OOK waveform while keeping the original detection window, timing robustness can be improved for the same SNR.

\section{5G LP-WUS Performance Evaluations}
In this section,  the benefits of shaped OOK designs discussed in Section V compared to rectangular OOK designs discussed in Section IV  are evaluated  
in the context of 3GPP simulations assumptions for LP-WUS.

OFDM  is configured with $\Nfft = 512 $, $\Ncp = 36$ and $f_{\rm sc}= 30 $ kHz.  The total signal bandwidth 
is 288 subcarriers, of which LP-WUS uses $144$ subcarriers inserted in the middle. Out of these, only $\Nsc =132$ are modulated, with guard bands of $N_{\rm GB}=6$ subcarriers on each side. Other subcarriers on adjacent channels are modulated by random QPSK symbols. 
LP-WUS carries $\Nbo=2$ information bits, Manchester-encoded into $\Nbit =4$ bits per OFDM symbol. 
3GPP TDL-C channel model   
is used with 300ns delay scaling and Rayleigh-faded taps according to 3km/h velocity at 2.6 GHz carrier frequency. 

At the receiver side,  both BPF and LPF are a 3rd order Butterworth filter with cutoff bandwidth matching the LP-WUS effective bandwidth of $\Nsc$ subcarriers.
ADC is with a 4-bit  quantization precision and a $3.86$ MHz sampling rate, corresponding to a downsampled transmitter sampling rate by a factor 4.  
The LR is assumed to have an automatic gain control, normalizing the received signal envelope  by its maximum peak and thus scaling it to the range $[0,\,1]$ before quantization.  
The reference time is assumed to be aligned with first channel tap. 
Manchester decoding is obtained by  segmenting the signal output from the ADC into $\Nbit$ segments.  The sample sums  for each pair of segments are then compared for info bit decision.

\begin{figure}[t] 
\vspace{-0.3cm}
\centering
\includegraphics[width=0.49\textwidth]{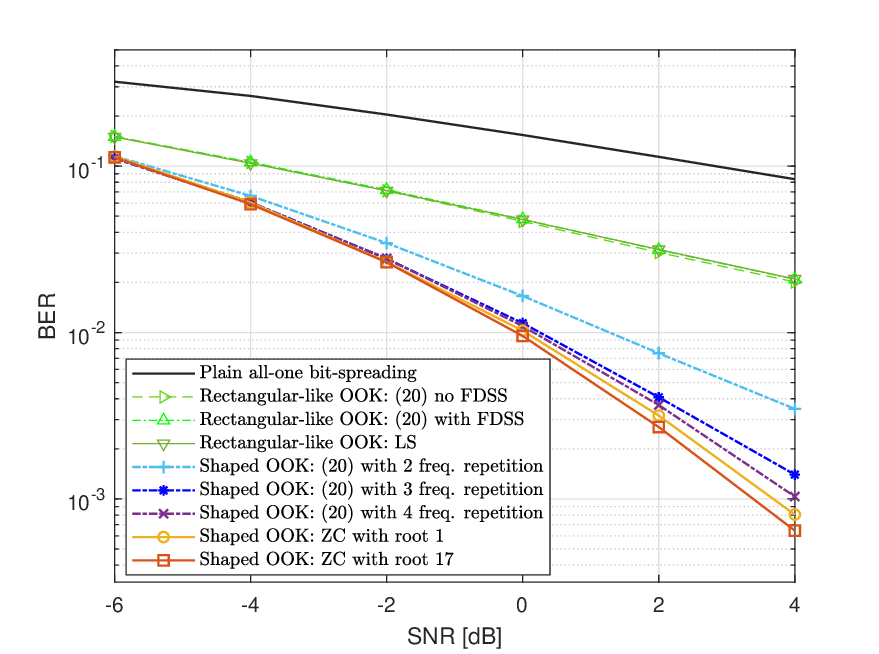}
\vspace{-0.7cm}
	\caption{BER performance comparison between different rectangular-like OOK designs and shaped-OOK designs with spread power sprectrum.}
	\label{fig:BER_DifferentWaveforms}
	\vspace{-0.5cm}
\end{figure}

\subsection{Performance Gain of OOK with FD Power Spreading over Rectangular-like OOK} 
Fig.~\ref{fig:BER_DifferentWaveforms} shows the bit error rate (BER) for different OOK designs under the assumption of perfect synchronization. 
The worst-performing curve corresponds to a plain DFT-s-OFDM modulation where bits are only trivially spread using an all-one sequence as $r[m]=1$, without additional processing (i.e. $L=\Phi=0$, $\Nsc =\Nsymb$ and no FDSS as $W[k]=1$). This method performs poorly because it fails to provide a good energy split between ON and OFF symbols, leading to an envelope shape similar  to the case of $\Phi=0$ in Fig.~\ref{fig:PhaseRampImpacta}. 
The rectangular-like OOK signals provide notable detection improvement, as they achieve an excellent energy split between ON and OFF symbols. Three different variations are considered: i) \eqref{eq:PhiFlat1} without FDSS; ii) \eqref{eq:PhiFlat1} with FDSS using  $\beta =4$ and a time-shifting of half a pulse;  and iii) LS waveform discussed in Sec. IV-C. 
These three rectangular-like OOK designs perform equivalently, showing that  envelope fluctuation inside ON symbols has no effect on performance.   
However, the performance of these designs is limited by their concentrated spectrum. 
By applying frequency repetition on top of rectangular-like design based on \eqref{eq:PhiFlat1} as discussed in Sec.~\ref{sec:FreqRep}, the performance  improves significantly  but saturates after 3 repetitions.  Here, 2, 3 and 4 repetitions are achieved by setting $\Nsymb=68$, $44$ and $32$, respectively. 
The best performance is obtained using a ZC overlaid sequence as discussed in Sec. V-A3. In this case, $L=\Phi=0$, $\Nsc =\Nsymb$ and there are no FDSS. The length of the ZC sequence is $\Nseg = 33$, with two roots, $u=1$ and $17$, considered. Both roots perform similarly, though the middle root $u=17$ performs slightly better. This can be attributed to a slightly better ON/OFF energy split with less OFF symbol fluctuation. 
 
\subsection{Performance Gain and Optimization of Concentrated OOK}  
\begin{figure}[t]
\vspace{-0.3cm}  
\centering
\subfigure[Comparison between normal and concentrated OOK. \label{fig:ConcentratedOOKa}]{\includegraphics[width=.44\textwidth]{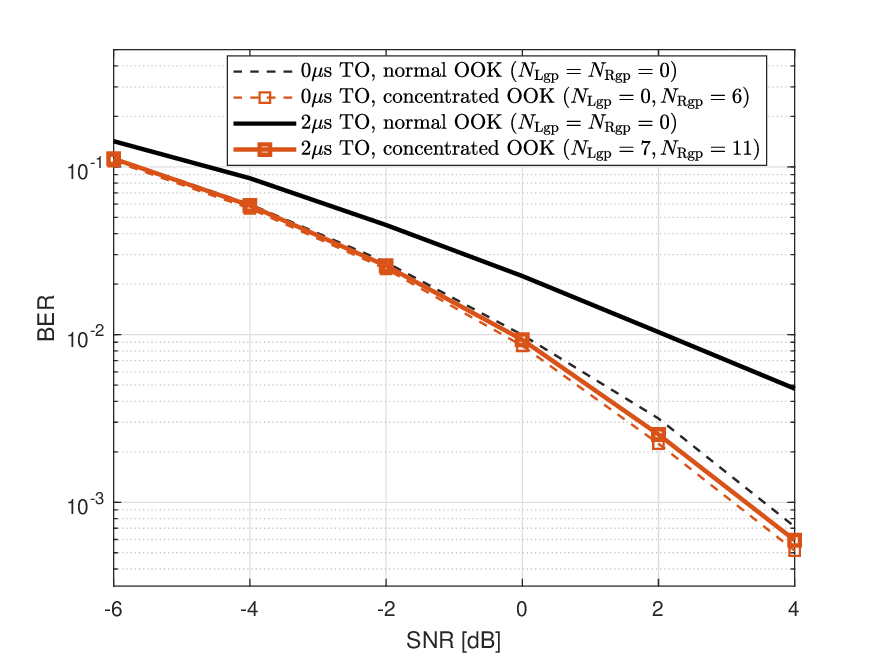}}
\subfigure[Different numbers of guard pulses.\label{fig:ConcentratedOOKb}]{\includegraphics[width=.24\textwidth]{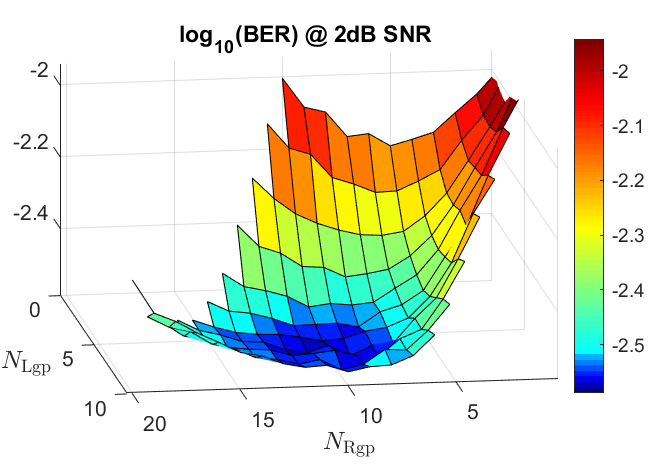}}
\subfigure[Concentration at Tx versus Rx.\label{fig:ConcentratedOOKc}]{\includegraphics[width=.24\textwidth]{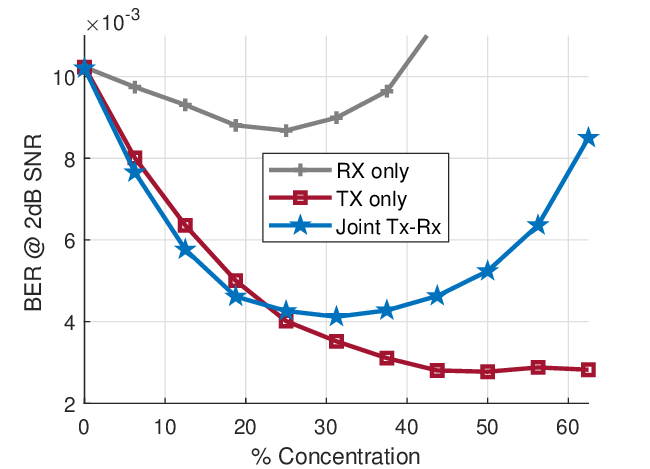}}
	\vspace{-0.3cm}
	\caption{BER evaluations of concentrated OOK with timing offset.}
	\label{fig:ConcentratedOOKAll}
	\vspace{-0.4cm}
	\end{figure}

In Fig~\ref{fig:ConcentratedOOKAll}, we evaluate the benefit of concentrated OOK for robustness against timing offset (TO). All curves use a ZC overlaid sequence with root $u=1$ whose length is   $\Nseg' =\Nseg -N_{\rm Lgp}- N_{\rm Rgp}$.  
In Fig~\ref{fig:ConcentratedOOKa}, with $\tau_{\rm err}= 0$ (i.e. no TO), we see that that increasing $N_{\rm Rgp}$  
up to $N_{\rm Rgp}=6$ can provide a fraction of dB gain. This corresponds to amortizing a channel spreading of $\tau_h = 1.5\,\mu$s, which accounts for absorbing $98\%$ of the average channel energy of the TDL-C channel model. 
With $\tau_{\rm err}=2\,\mu$s TO,  normal OOK experiences a 2dB SNR loss at $10^{-2}$ BER. Almost all of this performance drop can be recovered by concentrated OOK. 
Optimization on the number of guard pulses $(N_{\rm Lgp},N_{\rm Rgp})$ to minimize the BER variation at $2$dB SNR   is shown on Fig.~\ref{fig:ConcentratedOOKb}. It can be seen that the lowest BER (up to simulation precision) is obtained roughly in a zone $(5\leq N_{\rm Lgp}\leq 10,\,  7\leq N_{\rm Rgp}\leq 14)$, centered around $(N_{\rm Lgp},N_{\rm Rgp})= (7,11)$. This aligns with  the  anticipation in~\eqref{eq:NLgp} for compensating the maximum timing offset as here $\tau_{\rm err} \Nsc  f_{\rm sc} = 7.92$. However, the improvement from having an asymmetry $N_{\rm Lgp} \geq N_{\rm Lgp}$ shows diminishing returns as $N_{\rm Lgp}$ increases.  
Moreover, most BER gains can be achieved with fewer guard pulses, e.g. $(N_{\rm Lgp} ,\, N_{\rm Rgp})=(4,7)$, and additional concentration provides only marginal improvement. When concentration becomes too large, the BER starts to increase as a non-negligible portion of signal energy is lost during receiver downsampling.
Finally, in Fig.~\ref{fig:ConcentratedOOKc}, we considered concentrating the Rx detection window by comparing BER at $2$dB SNR and $\tau_{\rm err}=2\,\mu$s TO for three cases: i) only the TX signal design is concentrated, ii) only the RX window detection is concentrated, or iii) both are concentrated jointly. For the Tx signal-only concentration, $N_{\rm Lgp}=N_{\rm Rgp}$ and the original detection window is used as before, while with joint concentration the detection window is matching the concentrated ON symbol. Concentrating only the Rx detection window provides limited gain. Applying joint concentration offers marginal improvement for small concentration but becomes detrimental with larger concentrations. In these two cases,  a trade-off exists between timing error robustness and increased SNR. Ultimately, concentrating only the OOK signal appears to provide  the best performance.

\section{Conclusions} 
This paper studied a general framework for DFT-s-OFDM-based OOK modulation,  motivated by the recent adoption of OOK wake-up signals in modern wireless standards.  
Equivalent frequency-domain implementations were identified. Rectangular-like OOK signals with very flat OOK symbols and sharp edge transition were obtained, possibly with a zero-DC property, and recovering LS approach as a special case. Finally, frequency-domain and time-domain shaping methods useful for robustness against frequency-selective fading and timing error, respectively, were analyzed and evaluated. 
The results presented herein provide background and further insights into key signal design considerations of 3GPP LP-WUS study. 
5G-Advanced LP-WUS OOK signal design, whose details have been recently specified by 3GPP, is expected to be reused for other future IoT technologies as already done for 5G-Advanced Ambient-IoT.  

\section*{Acknowledgments}
The first author is grateful to Branislav M. Popovi\'{c} for valuable discussions on LP-WUS and his reviews of preliminary internal reports. The authors also thanks their 3GPP standardization colleagues Xue Yifan, Zhou Han, Li Qiang, Long Yi and Matthew Webb for fruitful discussions related to LP-WUS and Ambient-IoT. 

\appendix

\subsection{General Computation of Subcarrier Coefficients}
\label{App:GeneralComputation}
In a general setting of a bit spreading by $\Nseg = N_{x} / \Nbit$ where $N_{x}$ is  only restricted to be a factor of $ \Nbit$ such that $\Nseg$ is an integer, the modulation symbols $d[m]$ are then defined as in~\eqref{eq:d[m]} but for $m= 0,\ldots, N_{x}-1$, and the pre-processed subcarrier coefficients are  
$ D[k] = \sum_{m=0}^{N_{x}-1}d[m]  e^{- j \frac{2\pi}{N_{x}} km}$;  
the following result holds.
\begin{Lem} \label{Lem::D[K]general} 
\begin{equation}
D[k]  =  \sum_{l=0}^{\Nbit-1}  \tilde{R}_l[k]  b[ l]  e^{-\jrm\frac{2\pi}{\Nbit}k l }  \label{eq:D[k]asSumOfRl[k]}
\end{equation}
where the $\tilde{R}_l[k] = \sum_{m=0}^{\Nseg-1}   \tilde{r}_l[m] e^{-\jrm\frac{2\pi}{N_{x}}km}$ is the DFT of the $l$th overlaid sequence zero-padded to length $N_{x}$. 
\end{Lem}
\begin{proof} 
\begin{eqnarray}  
D[k]&=& \sum_{m=0}^{N_{x}-1} b\left[ l_m \right] r[m] e^{-\jrm\frac{2\pi}{N_{x}}km}  \nonumber\\
 &=&  \sum_{l=0}^{\Nbit-1} \sum_{m=l\Nseg}^{(l+1)\Nseg-1} b\left[ l_m \right] r[m] e^{-\jrm\frac{2\pi}{N_{x}}km}\nonumber\\
 &=&  \sum_{l=0}^{\Nbit-1}b\left[ l \right]  \!\!\!\sum_{m=l\Nseg}^{(l+1)\Nseg-1}   \!\!\!\tilde{r}_{ l_m} [m \; {\rm mod  }\; \Nseg] e^{-\jrm\frac{2\pi}{N_{x}}km}\nonumber\\
 &=& \sum_{l=0}^{\Nbit-1}b\left[ l \right] \sum_{m=0}^{\Nseg-1}   \tilde{r}_l[m] e^{-\jrm\frac{2\pi}{N_{x}}k(m+l\Nseg)}.  
\end{eqnarray}
which leads to~\eqref{eq:D[k]asSumOfRl[k]} since $\Nseg= N_{x}/\Nbit$. 
\end{proof}
An interpretation of the result above is that the sequence $\tilde{R}_l[k]$ creates  an OOK symbol for the $l$th bit, and these OOK symbols are time-shifted and multiplexed together within the OFDM symbol in~\eqref{eq:D[k]asSumOfRl[k]} according to the corresponding bit positions and values. 
 
Under the assumption of~\eqref{eq:SingleBitSpreading}:    
 \mbox{$\tilde{r}_l[m] =  e^{\jrm \Phi (m+ l\Nseg)} r_0[m] $}, DFTs of the overlaid sequences can be expressed as 
\begin{equation} 
\tilde{R}_l[k] 
 = e^{\jrm \Phi l\Nseg }  R_0\left(k-\frac{\Phi N_{x}}{2 \pi} \right) \label{eq:Rl[k]}
\end{equation}
where $R_0(f) = \sum_{m=0}^{\Nseg-1}   r_0[m] e^{-\jrm \frac{2\pi m f }{N_{x}}} $ 
is the interpolated DFT of the   zero-padded common overlaid sequence $r_0[n]$,  independent of the bit index $l$. 
Then,~\eqref{eq:D[k]asSumOfRl[k]} can be rewritten  as the product of two DFTs:  
\begin{equation} 
D[k] =  R_0\left(k-\frac{\Phi  N_{x}}{2 \pi} \right) B\left(k-\frac{\Phi  N_{x}}{2 \pi} \right) \label{eq:D=BR_Nx}
\end{equation}
 where $B(f)= \sum_{l=0}^{\Nbit-1}  b\left[ l\right] e^{-\jrm \frac{2\pi l f}{\Nbit}  }$ is the interpolated DFT coefficients of the bit sequence.

Lem.~\ref{Lem:D=BRCont} follows directly from~\eqref{eq:D=BR_Nx} with  $N_{x} = \Nsymb$.

Cor.~\ref{Corr:XPhaseRamp} follows from~\eqref{eq:D=BR_Nx} by taking $N_{x} = \Nsymb$ and the assumption of a common  all-one overlaid sequence, i.e. $r_0[m]= 1, \, \forall m $  as then by  using the exponential sum formula one gets 
 $R_0(f) = e^{-\jrm \pi f\frac{(\Nseg-1)}{N_{x}}  } \frac{\sin \left(\frac{\pi}{\Nbit} f\right)}{\sin \left(\frac{\pi}{N_{x}} f\right)}$ for $f\neq 0 \ { \rm mod }\ N_{x}$, and $R_0(f) = \Nseg = N_{x}/\Nbit$ otherwise.  

Lem.~\ref{Prop:LS} follows  similarly by identifying~\eqref{eq:DFTLS} with Lem.~\ref{Lem::D[K]general} and then~\eqref{eq:D=BR_Nx}  where $\Phi =0$, $N_{x} = \Nfft$, $r_0[m]= 1, \, \forall m $, and simplification of $R_0(f)$ by exponential sum formula. The final result follows by applying the shift and truncation of~\eqref{eq:XLS}.

\subsection{DFT of Manchester-Coded Bits}
\label{App:DFT_MCbits}
By incorporating Manchester coding~\eqref{eq:MC} in the computation of $\Bcal[k]$, \eqref{eq:Bcal[k]} becomes

\begin{eqnarray}
\Bcal[k] &=& \sum_{n=0}^{\Nbo-1} \overline{b_o[n]} e^{-\jrm 2n\frac{2 \pi k}{\Nbit}} +b_o[n] e^{-\jrm (2n+1) \frac{2 \pi k }{\Nbit}} \nonumber \\
&=& \sum_{n=0}^{\Nbo-1} e^{-\jrm 2n\frac{\pi k}{\Nbo}} \left( \overline{b_o[n]}  +b_o[n] e^{-\jrm \frac{ \pi k }{\Nbo}} \right).
\end{eqnarray}

Without loss of generality, we can write 
\begin{eqnarray}
 \overline{b_o[n]}  +b_o[n] e^{-\jrm \frac{ \pi k }{\Nbo}}  &=& 
 \overline{b_o[n]}e^{-\jrm \frac{ \pi k }{\Nbo}b_o[n]}   +b_o[n] e^{-\jrm \frac{ \pi k }{\Nbo}b_o[n]} \nonumber\\
&=& e^{-\jrm \frac{\pi k }{\Nbo}b_o[n]}  ( \overline{b_o[n]}  +b_o[n]) \nonumber \\
&=& e^{-\jrm \frac{ \pi k }{\Nbo}b_o[n]}
\end{eqnarray}
and so $\Bcal[k] =\sum_{n=0}^{\Nbo-1} e^{-\jrm \frac{\pi k}{\Nbo}(2n+b_o[n])}$.

\subsection{Flattening the OOK Waveforms}
\label{App:Flat}
The phase difference between two neighboring pulses of indices $m$ and  $(m+1)$ is
\begin{equation}
 \angle \frac{g_{m+1} [n]}{g_m [n]} = -\frac{2\pi L}{\Nsymb} + \angle \frac{h\left[n-\frac{\Nfft}{\Nsymb}(m+1)\right]}{h\left[n-\frac{\Nfft}{\Nsymb}  m \right]} . 
\end{equation}

Then, if the FDSS window is real and symmetric, $W[k] = W_R[k]$,  Lem.~1 of~\cite{KimTVT18} gives 
\begin{equation}
 \angle \frac{h\left[n-\frac{\Nfft}{\Nsymb}(m+1)\right]}{h\left[n-\frac{\Nfft}{\Nsymb}m \right]}=-\frac{(\Nsc-1)}{\Nsymb}\pi + \theta[n].
\end{equation}
This can be extended directly to FDSS windows of the form $W[k] = e^{-\frac{\jrm 2 \pi }{\Nfft} T_{\rm shift} k} W_R[k]$ as such phase shift only changes the overall time position of the pulses.  

The function $ \theta[n]= \{0 \text{ or } \pi\}$ corresponds to a sign difference between the real part of the pulses, and changes as a function of $n$. In the case of $\Nsymb= \Nsc$ and no FDSS, one can verify from~\eqref{eq:Dirichlet} that $\theta[n]=0$  for the samples between the two neighboring main lobe peaks, and as $\Nsymb$ is decreased, $\theta[n]=0$ in a smaller range of samples around the crossing of the main lobes. Numerical investigation confirms it is also the case with common FDSS windows. 
So in general $\theta[n]=0$ where pulse combining should be the most avoided, which leads to 
\begin{equation}
\angle  \frac{g_{m+1} [n]}{g_m [n]}= -\frac{\pi(2L+\Nsc-1)}{\Nsymb}  
\end{equation}
Finally,  by selecting $\Phi = \phi_{m+1}-\phi_{m}= \frac{\pi(2L+\Nsc-1)}{\Nsymb} $ we get $\angle  \frac{e^{\jrm \phi_{m+1}}g_{m+1} [n]}{e^{\jrm \phi_{m}}g_m [n]}= 0$. 

\subsection{LS Background}
\label{App:LS}
In~\cite{Mazloum20}, the subcarrier coefficients are obtained as~\cite[Eq. (8)]{Mazloum20} $\tilde{\xbf}_{LS}=(\tilde{\Fbf}^H \tilde{\Fbf})^{-1} \tilde{\Fbf}^H \bbf_m $ where $\bbf_m \in \{ 0, 1\}^{\Nfft \times 1}$ contains an upsampled version of the bits string $b[m]$ matching the OFDM IFFT size (i.e. $b_{\rm r}^{\rm LS}[n]$ in~\eqref{eq:DFTLS}), $\tilde{\Fbf} = \frac{1}{\Nfft} [e^{j\frac{2\pi m n}{\Nfft}}]_{\substack{n \in \Kcal_c \\0\leq m <\Nfft }
}$ is a truncated  $\Nfft \times \Nsc$ IDFT matrix whose columns have been restricted to the shifted indexes $\Kcal_c= (\Kcal-k_c) \; \text{mod} \, \Nfft $ with $\Kcal= f_0 + \{0,\ldots, \Nsc-1 \}$ corresponding to the LP-WUS subcarrier allocation and $k_c$ being the center index of the subset $\Kcal$. 
The value of the middle subcarrier $k_c$ is not explicitly given in~\cite{Mazloum20}, so we will conventionally take $k_c =f_0+ \lfloor \frac{\Nsc}{2} \rfloor$ 
 but  the alternative $k_c = f_0+ \lceil \frac{\Nsc}{2} \rceil$ could  also be selected.  

In 3GPP LP-WUS study~\cite{3GPPTR38.869}, this construction was often discussed as a specific type of precoding, different from DFT, while it is actually  simply a large DFT where only some outputs are retained. Indeed, the regularization matrix $(\tilde{\Fbf}^H \tilde{\Fbf})^{-1}$ is given following the generic least-squares formulation, however here this simplifies to $(\tilde{\Fbf}^H \tilde{\Fbf})^{-1}= \Nfft \Ibf_{\Nsc}$ such that this method reduces to $\tilde{\xbf}_{LS}=\bar{\Fbf} \bbf_m $ where $\bar{\Fbf} =  [e^{-\jrm\frac{2\pi m n}{\Nfft}}]_{\substack{0\leq n <\Nfft  \\ m \in \Kcal_c }}$ is a truncated $ \Nsc \times \Nfft $ DFT matrix where only  $\Nsc$ rows/outputs corresponding of subcarriers indices in $\Kcal_c $  are preserved. 
Note that  
the selected DFT output indices $\Kcal_c$ corresponds simply to the baseband subcarrier indices
$\Kcal_c=  \{-\lfloor  \frac{\Nsc}{2} \rfloor,\ldots, \lfloor  \frac{\Nsc-1}{2} \rfloor\} \; \text{mod}\, \Nfft $.   We encountered  
that this indexing notation in~\cite{Mazloum20} can be misunderstood by some readers.  
This is reflected in~\cite{Ericsson23} which is built on improving the LS approach from~\cite{Mazloum20} but  their description of it  
corresponds to use   
$\Kcal_c= \{0,\ldots, \Nsc-1 \}$  instead.   
Hence, the ``LS approach'' in~\cite{Ericsson23} has a much poorer rectangular shape than it should have.

\subsection{Proof of Lem.~\ref{Lem:E[|D_0[k]|^2]}}
\label{App:ProofLemPSD}
We start with the average power of the subcarrier coefficients given by $ E[|D_0[k]|^2] =  |R_0[k]|^2 E[|B[k]|^2]$ from \eqref{eq:D0=BR}. 

\subsubsection{Proof of \eqref{eq:E[|B[k]|^2]}}
The average $E[|B[k]|^2]$ is derived under the typical assumption of uniformly and independently distributed information bits. 
As given in Corr.~\ref{Corr:D=BR}, $B[k]  =   \Bcal[k\; {\rm mod }\; \Nbit]$ has a repetitive structure, so we need only to compute $E[|\Bcal[k]|^2]$.

\paragraph{With Manchester coding} $\Bcal[k]$ is given in \eqref{eq:B[K]MC} from Lem.~\ref{Lem::B[K]MC}.  
First, it is direct to verify $E[|\Bcal[0]|^2] = |\Bcal[0]|^2=\Nbo^2$. Then if $k\neq0$, by expansion we have 
\begin{equation}
|\Bcal[k]|^2 \!= \Nbo+ \!\!\sum_{\substack{n,m=0\\n\neq m} }^{\Nbo-1} e^{-\jrm \frac{2\pi k}{\Nbo}(n-m)} e^{-\jrm \frac{\pi k}{\Nbo}(b_o[n]-b_o[m])}.
\end{equation}
Averaging this term above, simplifications are obtained by observing that since $n\neq m$ in the sum the bits are independent, $E[e^{-\jrm \frac{\pi k}{\Nbo}b_o[n]}] = \frac{1}{2}(1+e^{-\jrm \frac{\pi k}{\Nbo}})$, and that $|1+e^{-\jrm \frac{\pi k}{\Nbo}}|^2 =2(1+ \cos \frac{\pi k}{\Nbo})$  by expansion and using Euler's formula. 
Also, we have 
\begin{eqnarray} 
\sum_{\substack{n,m=0\\n\neq m} }^{\Nbo-1} e^{-\jrm \frac{2\pi k}{\Nbo}(n-m)}   
&=&\sum_{n,m=0 }^{\Nbo-1} e^{-\jrm \frac{2\pi k}{\Nbo}(n-m)} - \Nbo \nonumber \\
&=& 0-\Nbo= -\Nbo.   \label{eq:DoubleSumIdentity}
\end{eqnarray}

Putting all that together, we get
\begin{eqnarray}
E[|\Bcal[k]|^2]\!\!\!\! &=& \!\!\!\! \Nbo+ \!\sum_{\substack{n,m=0\\n\neq m} }^{\Nbo-1} e^{-\jrm \frac{2\pi k}{\Nbo}(n-m)} E[e^{-\jrm \frac{\pi k}{\Nbo}(b_o[n]-b_o[m]) }]\nonumber \\
&=& 
\Nbo+ \frac{1}{2}\left(1+\cos \frac{\pi k}{\Nbo}\right) \sum_{\substack{n,m=0\\n\neq m} }^{\Nbo-1} e^{-\jrm \frac{2\pi k}{\Nbo}(n-m)} \nonumber \\ 
&=& \frac{\Nbo}{2}\left(1-\cos \frac{\pi k}{\Nbo}\right) 
\end{eqnarray}
which leads to \eqref{eq:E[|B[k]|^2]}.

\paragraph{Without Manchester coding} To derive the equivalent result without Manchester coding,  we average the power of~\eqref{eq:Bcal[k]} which is $|\Bcal[k]|^2 =\sum_{n,m=0}^{\Nbit-1} b[n]b[m]  e^{-\jrm \frac{2\pi k}{\Nbit}(n-m)} $.

For $k=0$,  
using the fact that for two independent bits we have $E[b[n]b[m]]= \frac{1}{4}$ and otherwise $E[b[n]^2]=\frac{1}{2}$, we have 
\begin{eqnarray}
E[|\Bcal[0]|^2] &=& \frac{\Nbit}{2} +\frac{1}{4}\sum_{\substack{n,m=0\\n\neq m} }^{\Nbit-1} 1 
 = \frac{\Nbit}{4} + \frac{1}{4}\sum_{n,m=0}^{\Nbit-1} 1\nonumber \\
&=&\frac{1}{4}\left(\Nbit^2 +\Nbit \right).
\end{eqnarray}

For the case $k\neq 0$, splitting again the case of dependent and independent bits  
and using a similar identity as~\eqref{eq:DoubleSumIdentity}, we have 
\begin{eqnarray}
E[|\Bcal[k]|^2] &=& \frac{\Nbit}{2}+ \frac{1}{4}\sum_{\substack{n,m=0\\n\neq m} }^{\Nbit-1} 
 e^{-\jrm \frac{2\pi k}{\Nbit}(n-m)} \nonumber \\
&=& \frac{\Nbit}{2}- \frac{\Nbit}{4} = \frac{\Nbit}{4}. 
\end{eqnarray}
Putting all together, we have 
\begin{equation}   
E[|B[k]|^2] = 
\begin{cases}    \frac{\Nbit(\Nbit+1)}{4}   \quad\quad   {\rm for} \;\;  ( k \,{\rm mod}  \, \Nbit) =0\\
    \frac{\Nbit}{4}  \quad {\rm otherwise}
		\end{cases}. \label{eq:E[B2]_NotMC}
\end{equation}

\subsubsection{Half Power Allocation on a Regular Comb}
For a $N$-length sequence $s=(s[0],\ldots, s[N-1])$, let us write its power as $\| s\|^2 = \sum_{n=0}^{N-1} |s[n]|^2$.  The $N$-DFT is a unitary operation if normalized by $1/\sqrt{N}$, so the DFT sequence $S[k] = \sum_{n=0}^{N-1}s[n] e^{-\jrm 2\pi \frac{nk}{N}}$ has power $\| S\|^2 = N \|  s\|^2$. 

With Manchester coding,  the bits string always satisfies $\|b \|^2 = \Nbo$, so the modulating symbols~\eqref{eq:d[m]} fulfills $\|d\|^2 =  \Nbo \| r_0 \|^2 $,  and the   sum  power of subcarrier coefficients~\eqref{eq:D[k]} is  $\|D\|^2 =  \Nbo \Nsymb \| r_0 \|^2 $.  
The share of this sum power on the comb $\Ccal= \{k'\Nbit\}_{k'=0}^{\Nseg-1}$ can be computed to be  
\begin{eqnarray}
\sum_{k\in \Ccal } |D_0[k]|^2\!\!\!\! &=& \!\!\!\!\! \sum_{k'=0}^{\Nseg-1} |R_0[k' \Nbit ]| ^2 |B[k' \Nbit ]|^2 \nonumber\\
 &=&\Nbo^2 \sum_{h=0}^{\Nseg-1}|R_0'[h ]| ^2  \nonumber\\
&=&  \Nbo^2 \Nseg \|r_0\|^2   
\end{eqnarray}
where the second equality follows from \eqref{eq:R_0[k]Interpolation} and  $|B[k' \Nbit ]|^2 = |\Bcal[0]|^2 =\Nbo^2$ due to the repetitive structure of $B[k]$. 
Since $\frac{\sum_{k\in \Ccal}  |D[k]|^2}{\|D\|^2} =  \frac{\Nbo \Nseg}{\Nsymb}=\frac{1}{2}$, this comb always carries half of the sum power.

Without Manchester coding, the \emph{average} power of the bits string is $E[\|b \|^2 ]= \Nbit/2$, so on average we also have $E[\|d\|^2] = \Nbit/2 \| r_0 \|^2 $,  and  $E[\|D\|^2] = \| r_0 \|^2 \Nsymb\Nbit/2  $.  
From~\eqref{eq:E[B2]_NotMC}, the share of average power on the comb  becomes
\begin{equation}
\sum_{k\in \Ccal}  E[|D_0[k ]|^2]  
=  \frac{\Nbit (\Nbit+1)}{4}  \Nseg \|r_0\|^2 
\end{equation}
and the ratio $\frac{\sum_{k\in \Ccal}  E[|D[k]|^2]}{E[\|D\|^2]} =  \frac{(\Nbit+1)}{2\Nbit} > \frac{1}{2}$.

\subsection{Proof of~\eqref{eq:R_0[k]Interpolation}}
\label{App:R0[k]}
We have 
$ R_0[k]= \sum_{l=0}^{\Nseg-1}r_0[l] e^{-\jrm  \frac{2\pi kl}{\Nsymb} }$, while by definition $ r_0[l]=\frac{1}{\Nseg} \sum_{h=0}^{\Nseg-1}R'_0[h] e^{\jrm  \frac{2\pi lh}{\Nseg} }$. Combining them and using the fact that $\Nsymb = \Nseg \Nbit$, we get 
\begin{equation}
R_0[k]= \frac{1}{\Nseg} \sum_{h=0}^{\Nseg-1} R'_0[h]
\sum_{l=0}^{\Nseg-1}  e^{-\jrm  \frac{ 2\pi l(k-\Nbit h)}{\Nsymb} }.
\end{equation}

If $k=k'\Nbit$ with $k'$ integer,  
\begin{equation}
R_0[k' \Nbit]= \frac{1}{\Nseg} \sum_{h=0}^{\Nseg-1} R'_0[h]
\sum_{l=0}^{\Nseg-1}  e^{-\jrm  \frac{ 2\pi l(k'-h)}{\Nseg} }.
\end{equation}
The inner sum above is equal to zero  if $h\neq k'$ and otherwise to $\Nseg$, so we have $R_0[k' \Nbit]=R'_0[k'] $ which gives the first line of~\eqref{eq:R_0[k]Interpolation}.

If $k\neq k' \Nbit$, using the sum exponential formula, 
\begin{equation}
\sum_{l=0}^{\Nseg-1}  \!\!\!\!e^{-\jrm  \frac{ 2\pi l(k-\Nbit h)}{\Nsymb}}
\!=\! \frac{\sin\left(\pi\frac{(k-\Nbit h)}{\Nbit} \right)}{\sin\left(\pi\frac{(k-\Nbit h)}{\Nsymb} \right)} e^{-\jrm \frac{\pi (\Nseg -1)}{\Nsymb}(k-\Nbit h)}
\end{equation}
and so we get the second line of~\eqref{eq:R_0[k]Interpolation}.

\addcontentsline{toc}{chapter}{Bibliography}
\bibliographystyle{IEEEtran} 
\bibliography{mybibfile}

\end{document}